\documentclass[12pt]{article}
\usepackage{amsfonts}
\usepackage{makeidx}
\usepackage{amssymb}
\usepackage{graphicx}
\usepackage{amsmath}
\usepackage[T1]{fontenc}

0cm \voffset -1.5cm \DeclareMathOperator\Tr{Tr}
\newtheorem{theorem}{Theorem}

\newtheorem{corollary}[theorem]{Corollary}

\newtheorem{adefinition}[theorem]{Definition}

\newtheorem{aexample}[theorem]{Example}

\newtheorem{lemma}[theorem]{Lemma}

\newtheorem{proposition}[theorem]{Proposition}
\newtheorem{aremark}[theorem]{Remark}
\newenvironment{remark}{\begin{aremark}\rm}{\end{aremark}}

\numberwithin{equation}{section} \numberwithin{theorem}{section}
\newenvironment{proof}[1][Proof]{\textbf{#1.} }{\ \rule{0.5em}{0.5em}}
\makeatother

\begin{document}

\title{On Non-Gaussian Limiting Laws for the Certain Statistics  of
the Wigner Matrices }
\author{A. Lytova \\
Mathematical Division\\
Institute for Low Temperatures\\
Kharkiv, Ukraine}
\date{}
\maketitle
 \begin{abstract}
We continue investigations
of our papers
\cite{Ly-Pa:08,Ly-Pa:09,Ly-Pa:11}, in which
there were proved CLTs for linear eigenvalue statistics
$\Tr\varphi (M^{(n)})$
and there were found the  limiting probability laws for the normalised matrix elements $\sqrt{n}\varphi_{jj}(M^{(n)})$ of
differential functions of  real symmetric Wigner
matrices $M^{(n)}$. Here we consider another spectral characteristic of Wigner
matrices, $\xi^{A} _{n}[\varphi ]=\Tr\varphi (M^{(n)})A^{(n)}$,  where $\{A^{(n)}\}_{n=1}^\infty$
is a  certain sequence of non-random matrices. We show  first that if  $M^{(n)}$
belongs to the Gaussian Orthogonal Ensemble (GOE), then $\xi^{A\circ} _{n}[\varphi ]$ satisfies CLT. Then we consider  Wigner matrices with i.i.d. entries possessing entire characteristic function and find the limiting probability law for $\xi^{A\circ} _{n}[\varphi ]$, which in general is not Gaussian.
\end{abstract}

\section{Introduction}
The asymptotic behavior of spectral characteristics of large
random matrices $ M^{(n)} $, when the size $ n $ of matrix tends
to infinity, is of the great interest in the random matrix theory.
One of the main questions under the study is the validity of CLT.  In the  last two decades there was
obtained a number of results on the CLT for linear eigenvalue
statistics $\Tr \varphi (M^{(n)})$ and other spectral
characteristics (see
\cite{An-Ze:06,Ba-Si:04,Ca:01,Ch:09,Co-Le:95,Di-Ev:01,Gu:02,Jo:98,Jo:82,
Ke-Sn:00,KKP,R-R-S:11,Sh:10,Si-So:97,Si-So:98,So:00} and
references therein). It was found that in many cases
fluctuations of various spectral characteristics of eigenvalues of
random matrix ensembles are asymptotically Gaussian
(see \cite{An-Ze:06,Ba-Si:04,Di-Ev:01,Gi:01,Jo:98,KKP,Pa:06b,R-R-S:11,Sh:10,Si-So:98,So:00}).
But the CLT is not always the case. Thus it was shown in
\cite{Pa:06b} that the CLT for linear eigenvalue statistics is not
necessarily valid for so called hermitian matrix models, for which
 in certain cases  appear non-Gaussian limiting laws.

Another example of
non-Gaussian limiting behavior
is presented in works  {\cite{Ly-Pa:11,O-R-Sh:11,R-R-S:11} dealing with the normalized individual  matrix elements
$\sqrt{n}\varphi_{jj}(M^{(n)})$ of functions
 of real symmetric Wigner random matrix. The particular case of
 matrix elements
$\sqrt{n}\varphi_{jj}(\widehat{M}^{(n)})$ with $\widehat{M}^{(n)}$
belonging to the GOE was considered
earlier in \cite{Ly-Pa:11}, where it was proved that
$\sqrt{n}(\varphi_{jj}(\widehat{M}^{(n)}))^\circ$
 satisfies the CLT. But  in {\cite{Ly-Pa:11,O-R-Sh:11,R-R-S:11} it was shown
that  in general
case of Wigner matrices the limiting probability law for
$\sqrt{n}(\varphi_{jj}(M^{(n)}))^\circ$ is not Gaussian but  the sum of the Gaussian
law and probability law of entries of $\sqrt{n}M^{(n)}$ modulo a certain
rescaling, and to obtain the CLT, one has to impose an integral
condition on the test function.

In particular, the fact that in contrast to the linear statistics of eigenvalues, individual  matrix elements
in general do not satisfy CLT reflects  influence of eigenvectors and gives some  information about asymptotic properties of eigenvectors. Indeed, in
the case of the Gaussian random matrices (GOE, null Wishart) the eigenvectors are rotationally
invariant and according to recent works \cite{Ba-Co:07,Er:10r,Le-Pe:09} the eigenvectors of the non-Gaussian random matrices (Wigner, sample covariance) are similar in several aspects to the eigenvectors of the Gaussian
random matrices. On the other hand, the results of \cite{Ly-Pa:09} and {\cite{Ly-Pa:11,O-R-Sh:11,R-R-S:11} imply that there are asymptotic properties of eigenvectors of the non-Gaussian random matrices which are different from those
for the Gaussian random matrices.

\medskip

This paper  continues the investigations of \cite{Ly-Pa:08,Ly-Pa:09,Ly-Pa:11}.
Here we consider  random variable
\begin{equation}
\xi^{A} _{n}[\varphi ]=\Tr\varphi (M^{(n)})A^{(n)},  \label{xif}
\end{equation}%
where  $\varphi $ is a smooth enough test-function and $ \{A^{(n)}\}_{n=1}^\infty$
is a sequence of $n\times n$ non-random matrix satisfying
\begin{align}
& \text{(i)}\;\;\lim_{n\rightarrow \infty }n^{-1}\Tr A^{(n)T}A^{(n)}=1,  \label{1} \\
& \text{(ii)}\;\;\exists \lim_{n\rightarrow \infty }n^{-1}\Tr A^{(n)}=T_{A}.
\label{TA}
\end{align}%
Let us make some examples:

\textit{1. Linear eigenvalue statistics.} If $A^{(n)}=I^{(n)}$, then
$T_A=1$ and
\begin{align}
&\xi^{A}_n[\varphi]=\Tr\varphi (M^{(n)}).\label{Nn} 
\end{align}

\textit{2. Matrix elements.} If $A^{(n)}_{lm}=\sqrt{n}\delta_{jl}\delta_{jm}$,
then $T_A=0$ and
\begin{align}
&\xi^{A}_n[\varphi]=\sqrt{n}\varphi_{jj}(M^{(n)}).  \label{fjj}
\end{align}

\textit{3. Bilinear forms.} If $A^{(n)}_{lm}=\sqrt{n}\eta_{l}\eta_{m}$, where
\begin{equation}
\eta^{(n)}=(\eta^{(n)}_1,...,\eta^{(n)}_n)^T, \quad \lim_{n\rightarrow\infty}\sum_{l=1}^n
(\eta^{(n)}_l)^2=1,\label{eta1}
\end{equation}
 then $T_A=0$ and
\begin{align}
&\xi^{A}_n[\varphi]=\sqrt{n}(\varphi(M^{(n)})\eta^{(n)},\eta^{(n)}).\label{bf}
\end{align}

\medskip

\noindent Here we  find the limiting probability law for  $\xi^{A}_n[\varphi]$ as $n\rightarrow\infty$.
Our main result is  Theorem \ref{t:clt}  below, where the limiting expression for
 characteristic function of   $\xi^{A\circ}_n[\varphi]$ is given and written
through
the cumulants of matrix entries and quantities depending on a sequence $ \{A^{(n)}\}_{n=1}^\infty$. Let us note that the corresponding theorems for
linear eigenvalue statistics (\ref{Nn}) and matrix elements (\ref{fjj}) of  \cite{Ly-Pa:08,Ly-Pa:09,Ly-Pa:11,P-R-Sh:11}
can be obtained from  Theorem  \ref{t:clt} as particular cases (however, under much stronger conditions).

The paper is organized as follows. Section \ref{s:t_means} contains definitions,
some known facts and technical means used throughout the paper.  In Section \ref{s:GOE} we consider the case of the Gaussian Orthogonal Ensemble (GOE) and prove CLT for $\xi^{A}_n[\varphi]$   (see \cite%
{Ly-Pa:09}   for the analogous statements for matrix elements). Then  we find the limiting variance, Sections \ref{s:cov},
and the limiting probability law, Sections \ref{s:main},   for  $\xi^{A}_n[\varphi]$ for the Wigner matrices.
Section \ref{s:aux} contains auxiliary results. We confine ourselves to real symmetric
matrices, although our results as well as the main ingredients of
proofs remain valid in the hermitian case with natural
modifications.

\textit{Convention}: We will use letter $c$ for an absolute constant that
does not depend on $j$, $k$, and $n$, and may be distinct on different
occasions.

\section{Definitions and Technical Means}\label{s:t_means}

To make the paper self-consistent, we present here several definitions and
technical facts
that will be often used below.
We start with  the
definition of the Wigner real symmetric matrix $M^{(n)}$, and put
\begin{equation}
M^{(n)}=n^{-1/2}W^{(n)},\quad W^{(n)}=\{W_{jk}^{(n)}\in \mathbb{R}%
,\;W_{jk}^{(n)}=W_{kj}^{(n)}\}_{j,k=1}^{n},  \label{MW}
\end{equation}%
where $\{W_{jk}^{(n)}\}_{1\leq j\leq k\leq n}$ are independent random
variables satisfying
\begin{equation}
\mathbf{E}\{W_{jk}^{(n)}\}=0,\quad \mathbf{E}\{(W_{jk}^{(n)})^{2}\}=w^{2}(1+%
\delta _{jk}).  \label{Wmom12}
\end{equation}%
The case of the Gaussian random variables
obeying (\ref{Wmom12}) corresponds to the GOE
(see e.g. \cite{Me:91}):
\begin{equation}
\widehat{M}^{(n)}=n^{-1/2}\widehat{W}^{(n)},\quad \widehat{W}^{(n)}=\{\;\widehat{W}%
_{jk}=\widehat{W}_{kj}\in \mathbb{R},\;\widehat{W}_{jk}\in \mathcal{N}%
(0,w^{2}(1+\delta _{jk}))\}_{j,k=1}^{n}.  \label{GOE}
\end{equation}%
Here for
simplicity sake we define Wigner matrix so that first two moments of its entries match  those of GOE. It can be shown that  if  $\mathbf{E}\{(W_{jj}^{(n)})^{2}\}=w^{2}w_{2}$,
then corresponding expressions for the limiting variance and characteristic
function have additional terms proportional to $(w_{2}-2)$ (see Remarks \ref{r:cov}
and \ref{r:main}).

We will assume in what follows additional conditions on distributions of $%
W_{jk}^{(n)}$, mostly in the form of existence of certain moments of $%
W_{jk}^{(n)}$, whose order will depend on the problem under study.

\medskip

The next proposition presents certain facts on Gaussian random
variables.
\begin{proposition}
\label{p:Nash} Let $\zeta=\{\zeta _{l}\}_{l=1}^{p}$ be independent
Gaussian random variables of zero mean, and $\Phi
:\mathbb{R}^{p}\rightarrow \mathbb{C}$ be
a differentiable function with polynomially bounded partial derivatives $%
\Phi _{l}^{\prime },\;l=1,...,p$. Then we have
\begin{equation}
\mathbf{E}{\mathbb{\{\zeta }}_{l}{\Phi (\zeta )}\}=\mathbf{\ E}{\mathbb{\{}{\zeta }%
}_{l}^{2}\}\mathbf{E}\mathbb{\{}{\Phi _{l}^{\prime }(\zeta
)}\},\;l=1,...,p, \label{diffga}
\end{equation}%
and%
\begin{equation}
\mathbf{Var}\{\Phi (\zeta) \}\leq \sum_{l=1}^{p}\mathbf{E}\{ \zeta _{l}^{2}\}\mathbf{E}%
\left\{ |\Phi _{l}^{\prime }(\zeta)|^{2}\right\} .  \label{Nash}
\end{equation}
\end{proposition}

The first formula is a version of the integration by parts. The
second is a version of the Poincar\'{e} inequality (see e.g.
\cite{Ba:98}).
Formula (\ref{diffga}) is a particular case of more general formula. To
write it we recall some definitions. If a random variable $\zeta $ has a
finite $p$th absolute moment, $p\geq 1$, then we have the expansions
\begin{equation*}
\mathbf{E}\{e^{it\zeta }\}=\sum_{j=0}^{p}\frac{\mu _{j}}{j!}%
(it)^{j}+o(t^{p}),\quad t\rightarrow 0,
\end{equation*}%
and
\begin{equation}
\log \mathbf{E}\{e^{it\zeta }\}=\sum_{j=0}^{p}\frac{\kappa _{j}}{j!}%
(it)^{j}+o(t^{p}),\quad t\rightarrow 0,  \label{lt}
\end{equation}%
where $"\log "$ denotes the principal branch of logarithm. The coefficients
in the expansion of $\mathbf{E}\{e^{it\zeta }\}$ are the moments $\{\mu _{j}\}$ of $\zeta $, and the
coefficients in the expansion of $\log \mathbf{E}\{e^{it\zeta }\}$ are the cumulants $\{\kappa _{j}\}$ of $%
\zeta $. For small $j$ one easily expresses $\kappa _{j}$ via $\mu _{1},\mu
_{2},\dots ,\mu _{j}$. In particular, if $\mu_1=0$, then
\begin{align}
\kappa _{1} =0,\quad \kappa _{2}=\mu _{2}=\mathbf{Var}%
\{\zeta \},\quad \kappa _{3}=\mu _{3},\quad
\label{cums}
\kappa _{4} =\mu _{4}-3\mu _{2}^{2},\;  ...
\end{align}%
 We have  \cite%
{KKP,Ly-Pa:08}:

\begin{proposition}
\label{l:difgen} (i) Let $\zeta $ be a random variable such that $\mathbf{E}\{|\zeta
|^{p+2}\}<\infty $ for a certain non-negative integer $p$. Then for any
function $\Phi :\mathbb{R}\rightarrow \mathbb{C}$ \ of the class $C^{p+1}$
with bounded partial derivatives $\Phi ^{(l)}$, $l=1,..,p+1$, we have
\begin{equation}
\mathbf{E}\{\zeta \Phi (\zeta )\}=\sum_{l=0}^{p}\frac{\kappa _{l+1}}{l!}\mathbf{E%
}\{\Phi ^{(l)}(\zeta )\}+\varepsilon _{p},  \label{difgen}
\end{equation}%
where
\begin{equation}
|\varepsilon _{p}|\leq C_{p}\mathbf{E}\{|\zeta |^{p+2}\}\sup_{t\in \mathbb{R}%
}|\Phi ^{(p+1)}(t)|,\,\,\,C_{p}\leq \frac{1+(3+2p)^{p+2}}{(p+1)!}.
\label{b3}
\end{equation}
(ii) If the characteristic function $\mathbf{E}\{e^{it|\zeta| }\}$ is entire, and $%
\Phi \in C^\infty$, then
\begin{equation}
\mathbf{E}\{\zeta \Phi (\zeta )\}=\sum_{l=0}^{\infty}\frac{\kappa _{l+1}}{l!}%
\mathbf{E}\{\Phi ^{(l)}(\zeta )\}  \label{difinf}
\end{equation}
provided that for some $a>0$
\begin{equation}
|\mathbf{E}\{\Phi ^{(l)}(\zeta )\}|\leq a^l,  \label{al}
\end{equation}
and for some $R=ca,$ $c>1$,
\begin{equation}
\sum_{l=0}^{\infty}\frac{|\kappa _{l+1}|R^l}{l!}<\infty.  \label{kap<}
\end{equation}
\end{proposition}

\medskip

\noindent
Here is a simple "interpolation" corollary  showing the mechanism
of proximity of expectations with respect to the probability law
of an arbitrary random variable and the Gaussian random variable
with the same first and second moments. Its multivariate version
will be often used below.

\begin{corollary}
\label{c:difgen}  Let $\widehat{\zeta }$
be the Gaussian random variable, whose first and second moments
coincide with those of given random variable $\zeta $. Then:

\noindent (i) We have
under conditions of Proposition \ref{l:difgen} (i):
\begin{equation}
\mathbf{E}_{\zeta }\{\Phi (\zeta )\}-\mathbf{E}_{\widehat{\zeta }}\{\Phi (\widehat{%
\zeta })\}=\sum_{l=2}^{p}\frac{\kappa
_{l+1}}{2l!}\int_{0}^{1}\mathbf{E}\{\Phi ^{(l+1)}(\zeta
(s))\}s^{(l-1)/2}ds+\varepsilon _{p}^{\prime },  \label{difint}
\end{equation}%
where the symbols $\mathbf{E}_{\zeta }\{...\}$ and $\mathbf{E}_{\widehat{\zeta }%
}\{...\}$ denote the expectation with respect to the probability
law of $\zeta
$ and $\widehat{\zeta }$, $\{\kappa _{j}\}$ are the cumulants of $\zeta $, $%
\mathbf{E}\{...\}$ denotes the expectation with respect to the
product of probability laws of $\zeta $ and $\widehat{\zeta }$,
\begin{align}
&\zeta (s)=s^{1/2}\zeta +(1-s)^{1/2}\widehat{\zeta },\quad 0\leq s\leq
1,
\label{xint} \\
&|\varepsilon _{p}^{\prime }|\leq C_{p}\mathbf{E}\{|\zeta
|^{p+2}\}\sup_{t\in \mathbb{R}}|\Phi ^{(p+2)}(t)|,  \label{edprim}
\end{align}%
and $C_{p}$ satisfies (\ref{b3}).

\noindent (ii) We have
under conditions of Proposition \ref{l:difgen} (ii):
\begin{equation}
\mathbf{E}_{\zeta }\{\Phi (\zeta )\}-\mathbf{E}_{\widehat{\zeta }}\{\Phi (\widehat{%
\zeta })\}=\sum_{l=2}^{\infty}\frac{\kappa
_{l+1}}{2l!}\int_{0}^{1}\mathbf{E}\{\Phi ^{(l+1)}(\zeta
(s))\}s^{(l-1)/2}ds  \label{inf}
\end{equation}
with $\zeta(s)$ given above.
\end{corollary}\medskip

The next proposition presents simple facts of linear algebra

\begin{proposition}
\label{p:Duh} Let $M$ and $M^{\prime }$ be $n\times n$ matrices, and $t\in
\mathbb{R}$. Then we have the following:

\begin{enumerate}
\item[(i)] the Duhamel formula
\begin{equation}
e^{(M+M^{\prime })t}=e^{Mt}+\int_{0}^{t}e^{M(t-s)}M^{\prime }e^{(M+M^{\prime
})s}ds,  \label{Duh}
\end{equation}

\item[(ii)] if for a real symmetric $n\times n$ matrix $M^{(n)}$ we put
\begin{equation}
U(t)=U^{(n)}(t):=e^{itM^{(n)}},\;t\in \mathbb{R},  \label{U}
\end{equation}%
then $U(t)$ is a symmetric unitary matrix satisfying
\begin{equation}
U(t_{1})U(t_{2})=U(t_{1}+t_{2}),\quad ||U(t)||=1,\quad
\sum_{j=1}^{n}|U_{jk}(t)|^{2}=1,  \label{norU}
\end{equation}

\item[(iii) ] if $D_{lm}=\partial /\partial M_{lm}$, then
\begin{equation}
D_{lm}U_{ab}(t)=i\beta _{lm}\left (U_{al}\ast U_{bm}+U_{bl}\ast
U_{am}\right)(t),  \label{ParU}
\end{equation}%
where
\begin{align}
&\beta _{lm}=(1+\delta _{lm})^{-1}=1-\delta_{lm}/2,  \label{beta}
\end{align}%
the symbol "$\ast $" is defined in Proposition \ref{p:Four} (ii), and
\begin{equation}
|D_{lm}^{p }U_{ab}(t)|\leq c'_{p}|t|^{p},\quad c' _p=2^p/p!,  \label{ocdlu}
\end{equation}
\item[(iv)] if $A^{(n)}$  is an $n\times n$ matrix and $\xi^{A}_n(t_{})=\mathrm{Tr}A^{(n)}U(t)$, then
\begin{align}
  D_{lm}(A^{(n)}U)_{ab}(t)&=i\beta _{lm}\left  ((A^{(n)}U)_{al}\ast U_{bm}+U_{bl}\ast
(A^{(n)}U)_{am}\right)(t),\label{DAU} \\
 D_{lm}\xi^{A}_n(t)&=i\beta _{lm}(U\ast
C^{(n)}U)_{lm}(t),\quad C^{(n)}=A^{(n)}+A^{(n)T},\label{Dxi}\\
D^{2}_{lm}\xi^{A}_n(t)&=-\beta^{2} _{lm}\big(U_{ll}\ast
(U\ast
C^{(n)}U)_{mm}+U_{mm}\ast
(U\ast
C^{(n)}U)_{ll}\notag\\
&\hspace{2cm}+2U_{lm }\ast
(U\ast
C^{(n)}U)_{lm}\big)(t),\label{D2xi}
\\
D_{lm}(U\ast A^{(n)}U)_{jk}(t)&=i\beta _{lm}\big(U_{jl}\ast
(U\ast A^{(n)}U)_{mk}+U_{jm}\ast
(U\ast A^{(n)}U)_{lk}\label{DUAU}
\\
&\hspace{2cm}+U_{lk}\ast
(U\ast
A^{(n)}U)_{jk}+U_{mk}\ast
(U\ast
A^{(n)}U)_{jl}\big)(t),\notag\\
D_{lm}(U\ast
A^{(n)}U)_{lm}(t)&=i\beta _{lm}\big(U_{ll}\ast
(U\ast
A^{(n)}U)_{mm}+U_{mm}\ast
(U\ast
A^{(n)}U)_{ll}\notag\\
&\hspace{2cm}+2U_{lm }\ast
(U\ast
A^{(n)}U)_{lm}\big)(t).\label{DUAUlm}
\end{align}
\end{enumerate}
\end{proposition}

\medskip
\noindent It follows from the above that if  $A^{(n)}$ satisfy (\ref{1}) -- (\ref{TA}),
and
\begin{align}
 C_A:\;\Tr A^{(n)T}A^{(n)}\leq C_{A} n,\; \forall n\in \mathbb{N}, \label{CA}
\end{align}
then
\begin{align}
&  |(A^{(n)}U^{(n)})_{lm}|\leq(A^{(n)T}A^{(n)})_{ll}^{1/2}\leq O(n^{1/2}),\label{AUOn} \\
&  |(U^{(n)}A^{(n)}U^{(n)})_{lm}|\leq(\Tr A^{(n)T}A^{(n)})^{1/2}\leq C_A n^{1/2},\label{UAUOn} \\
&  \sum_{l,m=1}^n|(U^{(n)}A^{(n)}U^{(n)})_{lm}|^2=\Tr
A^{(n)T}A^{(n)}= O(n^{}),\label{UAUlm}
\\
&  \sum_{m=1}^n|(U^{(n)}A^{(n)}U^{(n)})_{mm}|^2\leq O(n),\label{UAUmm}
\end{align}
and
\begin{align}
& |\xi^{A}_n(t)|\leq(n\Tr A^{(n)T}A^{(n)})^{1/2}= O(n),\label{xiOn}
\\
& |D_{lm}^{p }\xi^{A}_n(t)|\leq  n^{1/2}c_{p}|t|^{p},\quad c_p=C_A 2^{p+1}/p!,  \label{dxi<}
\end{align}
as $n\rightarrow\infty$.
\vskip0.5cm

\noindent At last we need the generalized Fourier transform, in fact the $\pi /2$
rotated Laplace transform (see e.g. \cite{Ti:86}, Sections 1.8-9 for its
definition).

\begin{proposition}
\label{p:Four} Let $f:\mathbb{R}_{+}\rightarrow \mathbb{C}$ be a locally
Lipshitzian and such that for some $\delta >0$%
\begin{equation}
\sup_{t\geq 0}e^{-\delta t}|f(t)|<\infty ,  \label{supde}
\end{equation}%
and let $\widetilde{f}:\{z\in \mathbb{C}:\Im z<-\delta \}\rightarrow \mathbb{%
C}$ be its \emph{generalized Fourier transform}
\begin{equation}
\widetilde{f}(z)=i^{-1}\int_{0}^{\infty }e^{-izt}f(t)dt.  \label{Fur}
\end{equation}%
The inversion formula is given by
\begin{equation}
f(t)=\frac{i}{2\pi }\int_{{L}}e^{izt}\widetilde{f}(z)dz,\;t\geq 0,
\label{Furinv}
\end{equation}%
where ${L}=(-\infty -i\varepsilon ,\infty -i\varepsilon )$, $\varepsilon
>\delta , $ and the principal value of the integral at infinity is used.

Denote for the moment the correspondence between functions and their
generalized Fourier transforms as $f\leftrightarrow \widetilde{f}$. Then we
have:

\begin{enumerate}
\item[(i)] $\quad \int_{0}^{t}f(\tau )d\tau \leftrightarrow (iz)^{-1}%
\widetilde{f}(z);$

\item[(ii)] $\quad \int_{0}^{t}f_{1}(t-\tau )f_{2}(\tau )d\tau :=(f_{1}\ast
f_{2})(t)\leftrightarrow i \widetilde{f_{1}}(z)\widetilde{f_{2}}(z);$

\item[(iii)] if $P$, $Q$, and $R$ are differentiable, and $R(0)=0$, then the
equation
\begin{equation}
P(t)+\int_{0}^{t}dt_{1}%
\int_{0}^{t_{1}}Q_{{}}(t_{1}-t_{2})P(t_{2})dt_{2}=R(t),\;t\geq 0,
\label{intrel}
\end{equation}%
has a unique differentiable solution
\begin{equation}
P(t)=-\int_{0}^{t}T_{{}}(t-t_{1})R^{\prime }(t_{1})dt_{1},  \label{solut}
\end{equation}%
where
\begin{equation}
T\leftrightarrow (z+\widetilde{Q})^{-1}  \label{TtSt}
\end{equation}%
provided by
\begin{equation}
z+\widetilde{Q}(z)\neq 0,\;\Im z<0.  \label{condQ1}
\end{equation}%
\end{enumerate}
\end{proposition}
Applying the generalized Fourier transform we prove the lemma, which will be often used in what follows:
\begin{lemma}\label{l:eq} Consider
\begin{equation}
{v}(t)=\int
_{-2w}^{2w}e^{it\lambda}\rho_{sc}(\lambda)d\lambda,  \label{vt}
\end{equation}
where $\rho_{sc}$ is the density of the semicircle law
\begin{equation}
\rho _{sc}(\lambda ) =(2\pi w^{2})^{-1}\big((4w^{2}-\lambda ^{2})^{1/2}\mathbb{I}_{[-2w^2,2w^2]}\big)^{1/2}.
\label{rhosc}
\end{equation}
Then  unique differentiable solutions
of integral equations
\begin{align}
&F_{1}(t)+w^{2}\int_{0}^{t}dt_{1}%
\int_{0}^{t_{1}}v(t_{1}-t_{2})F_{1}(t_{2})dt_{2}=1,\label{F1e}
\\
&F_{2}(t_{1},t_{2})+w^{2}\int_{0}^{t_{1}}dt_{3}%
\int_{0}^{t_{3}}v(t_{3}-t_{4})F_{2}(t_{4},t_{2})dt_{4}\notag
\\
&\hspace{3cm}=-w^{2}\int_{0}^{t_{1}}dt_{3}%
\int_{0}^{t_{2}}v(t_{2}-t_{4})v(t_{3}+t_{4})dt_{4},\label{F2e}
\\
&F_{3}(t_{1},t_{2})+2w^{2}\int_{0}^{t_{1}}dt_{3}%
\int_{0}^{t_{3}}v(t_{3}-t_{4})F_{3}(t_{4},t_{2})dt_{4}\notag
\\
&\hspace{3cm}=-2w^{2}t_{2}\int_{0}^{t_{1}}v(t_{2}+t_{3})dt_{3},\label{F3e}
\end{align}%
 are given by
 \begin{align}
&F_{1}(t)=v(t),\label{F1}
\\
&F_{2}(t_{1},t_{2})=v(t_{1}+t_{2})-v(t_{1})v(t_{2}),\label{F2}
\\
&F_{3}(t_{1},t_{2})=
\frac{1}{2\pi ^{2}}\int_{-2w}^{2w}\; \int_{-2w}^{2w}\frac{\Delta e(t_{1})\Delta e(t_{2})}{(\lambda_{1} -\lambda_{2} )^{2}}\frac{4w^2-\lambda_1\lambda_2}{%
\sqrt{4w^2-\lambda_1^2}\sqrt{4w^2-\lambda_2^2}} d\lambda_{1}d\lambda _{2},\label{F3}
\end{align}
where
we denote
\begin{equation}
\Delta e(t)=e^{it\lambda_{2}}-e^{it\lambda_{1}}.\label{e12}
\end{equation}
\end{lemma}
\begin{proof}
Note first that in fact the generalized Fourier transform $\widetilde{v}$
of $v$ is the Stiltjes transform of the semicircle law density (\ref{rhosc}):    \begin{equation*}
\widetilde{v}(z)=\int_{\mathbb{R}}
\frac{\rho_{sc}(\lambda)d\lambda}{\lambda-z},  \quad \Im z \neq 0,
\end{equation*}
so that
\begin{align}
&w^2\widetilde{v}(z)^2+z\widetilde{v}(z)+1=0, \label{eqvz}
\\
&\widetilde{v}(z)=(2w^2)^{-1}(\sqrt{z^2-4w^2}-z),\notag
\end{align}%
where   $\sqrt{z^2-4w^2}$ is defined by the asymptotic $\sqrt{z^2-4w^2}=z+O(z^{-1})$,
  $z\rightarrow\infty$. Denote $\widetilde{F}_1(z)$, $\widetilde{F}_j(z,t_2)$, $j=2,3$ the generalized Fourier transforms of $F_j$, $j=1,2,3$. We have for
$\widetilde{F}_1$:
 \begin{align*}
&\widetilde{F}_1(z)(1+w^2\widetilde{v}(z)z^{-1})=z^{-1},
\end{align*}
hence, $\widetilde{F}_1(z)=\widetilde{v}(z)$ (see (\ref{eqvz})), and we
get (\ref{F1}). We also have for
$F_2$ by (\ref{solut}) with $T=-v$ and $R'(t)=-w^{2}%
\int_{0}^{t_{2}}v(t_{2}-t_{4})v(t+t_{4})dt_{4}$:
 \begin{align*}
&F_{2}(t_{1},t_{2})=-w^{2}\int_{0}^{t_{1}}v(t_{1}-t_{3})dt_{3}%
\int_{0}^{t_{2}}v(t_{2}-t_{4})v(t_{3}+t_{4})dt_{4}
\end{align*}
and after some calculations one can get
\begin{align*}
F_{2}(t_{1},t_{2})& =
\frac{1}{2}
\int_{-2w}^{2w}\int_{-2w}^{2w}\Delta e(t_1)\Delta e(t_2)\rho _{sc}(\lambda _{1})\rho _{sc}(\lambda
_{2})d\lambda _{1}d\lambda _{2}
\\
&=v(t_1+t_2)-v(t_1)v(t_2),
\end{align*}
where $\Delta e$ is defined in (\ref{e12}).
Consider now equation (\ref{F3e}). In this case we have for $Q$ of (\ref{intrel})
 \begin{equation*}
z+{\widetilde{Q}}(z)=\sqrt{z^2-4w^2}\neq 0, \; \Im z <0,
\end{equation*}%
so that condition (\ref{condQ1}) is fulfilled.
  This yields for $T$ of (\ref{TtSt})
\begin{equation}\label{Ttga}
T(t)=-\frac{1}{2\pi i}\int_L\frac{e^{izt}dz}{\sqrt{z^2-4w^2}}=
-\frac{1}{\pi}\int_{-2w}^{2w}\frac{e^{i\lambda t}d\lambda}{\sqrt{4w^2-\lambda^2}}
\end{equation}
(we replaced the integral over $L$ by the integral over the edges of the cut
$[-2w,2w]$).
This  and (\ref{solut}) lead to
\begin{align*}
F_{3}(t_1,t_2)&=\frac{2it_{2}}{2\pi ^{2}}
\int_{-2w}^{2w}e^{it_{2}\lambda_{2}}\sqrt{4w^{2}-\lambda_{2} ^{2}}%
d\lambda_{2} \int_{-2w}^{2w}\frac{e^{it_{1}\lambda_{2}}-e^{it_{1}\lambda_{1}}}{\sqrt{4w^{2}-\lambda_{1}^2}(\lambda_{2}-\lambda_{1} )}d\lambda_{1}\\
&=\frac{1}{\pi ^{2}}\int_{-2w}^{2w}\frac{e^{it_{1}\lambda_{1}}}{\sqrt{4w^{2}-\lambda_{1}^2}}d\lambda_{1}\int_{-2w}^{2w}\frac{\partial }{\partial\lambda_{2}}(e^{it_{2}\lambda_{2}}-e^{it_{2}\lambda_{1}})\frac{\sqrt{4w^{2}-\lambda^{2}_2}}{\lambda_{1}-\lambda_{2}}%
d\lambda_{2},
\end{align*}
where we used
\begin{equation*}
\int_{-2w}^{2w}\frac{d\lambda_{1}}{\sqrt{4w^{2}-\lambda_{1}^2}(\lambda_{2}-\lambda_{1} )}=0,\quad
|\lambda_{2}|\leq 2w.
\end{equation*}
Integrating by
parts with respect to $\lambda_{2} $, and writing then the half-sum of the obtained expression and the expression with interchanged
variables $\lambda_{1} \longleftrightarrow \lambda_{2} $, we get (\ref{F3}).
\end{proof}

\section{The GOE case}\label{s:GOE}

Denote by
\begin{equation}
F[\varphi](t)=\frac{1}{2\pi }\int e^{-it\lambda }\varphi(\lambda )d\lambda   \label{FT}
\end{equation}
the standard Fourier
transform of $\varphi$. Writing the Fourier inversion formula%
\begin{equation}
\varphi(\lambda )=\int e^{i\lambda t}F[\varphi](t)dt
\label{invFT}
\end{equation}%
and using  the spectral theorem
for symmetric  matrices, we obtain%
\begin{equation}\label{xifxit}
\xi^{A}_n[\varphi]=\int\xi^{A}_n(t)
F[\varphi](t)dt,
\end{equation}%
where $\xi^{A}_n(t)$ is a particular case of $\xi^{A}_n[\varphi]$ corresponding
to $\varphi(\lambda)=e^{it\lambda}$:
\begin{align}\label{xint}
&\xi^{A}_n(t)=\mathrm{Tr}A^{(n)}U(t),
\\
&U(t)=U^{(n)}(t):=e^{itM^{(n)}}\notag
\end{align}%
(see also (\ref{U})). Denote
\begin{equation}
v_{n}(t)=n^{-1}\xi^{I}_n(t)=n^{-1}\mathrm{Tr}U(t).  \label{vn}
\end{equation}
Since
for any bounded continuous $\varphi $
 \begin{equation*}
\lim_{n\rightarrow\infty}n^{-1}\mathbf{E}\{\Tr \varphi(M^{(n)})\}=\int_{-2w}^{2w}
\varphi(\lambda)\rho_{sc}(\lambda)d\lambda,
\end{equation*}
where  $M^{(n)}$ is   Wigner matrix and
 $\rho_{sc}$ is the density of the semicircle law (\ref{rhosc}) (see e.g. \cite{Pa:05} and references therein), then
we have\begin{equation}
\lim_{n\rightarrow\infty}\mathbf{E}\{v_n(t)\}=v(t),
\label{vnvt}
\end{equation}
where $v$ is defined in (\ref{vt}).

In this chapter we consider $\xi^{A} _{n}[\varphi ]$  corresponding to the
 GOE matrix  $M^{(n)}=\widehat{M}^{(n)}$.
In view of the orthogonal invariance of GOE probability measure
we have\begin{equation}
\mathbf{E}\{U_{jk}(t)\}=\delta _{jk}\mathbf{E}\{v_n(t)\}, \label{EUjk}
\end{equation}
so that
\begin{equation*}
n^{-1}\mathbf{E}\{{\xi}^{A}_n(t)\}=\mathbf{E}\{v_n(t)\}n^{-1}\Tr
A^{(n)}
\end{equation*}
and
\begin{equation}
\lim_{n\rightarrow\infty}n^{-1}\mathbf{E}\{{\xi}^{A}_n(t)\}= T_{A}\cdot v(t),\quad
\label{xiTA}
\end{equation}
where $T_A$ is defined in (\ref{TA}). We also have:
\begin{lemma}
\label{l:Var}
\noindent Let $\widehat{M}^{(n)}$ be the GOE matrix (\ref{GOE}). Denote
\begin{equation}
\xi _{n}^{A\circ }[\varphi ]=\xi^{A} _{n}[\varphi ]-\mathbf{E}\{\xi^{A} _{n}[\varphi
]\}.  \label{xic}
\end{equation}
Then for any test-function $\varphi :\mathbb{R\rightarrow C}$,
whose Fourier transform
(\ref{FT}) satisfies the condition
\begin{equation}
\int(1+|t|)^{}|F[\varphi ](t)|dt<\infty,  \label{condF}
\end{equation}%
we have the bound
\begin{align}
\mathbf{Var}\{\xi^{A}_n[\varphi]\}:&=\mathbf{E}\{|\xi^{A\circ}_n[\varphi]|^2\}
\leq c\Big( \int (1+|t|)^{}|F[\varphi ](t)|dt\Big) ^{2}.
\label{VarFG}
\end{align}
\end{lemma}
\begin{proof}
It follows from Poincar\'{e} inequality (\ref{Nash}) and (\ref{Dxi})
that
\begin{align*}
\mathbf{Var}\{\xi^{A}_n(t)\} &\leq\frac{w^{2}}{n} \sum_{1\leq l\leq m \leq
n} \beta^{-1}_{lm}\mathbf{E}\{|D_{lm}\xi^{A}_n(t)|^2\}
\\
&\leq\frac{2w^{2}}{n} \sum_{ l, m=1}^n \mathbf{E}\{|(U*A^{(n)}U)_{lm}(t)|^2\}=\frac{2w^{2}|t|^2}{n}\Tr
AA^{(n)T},
\end{align*}
so that
\begin{align}
\mathbf{Var}\{\xi^{A}_n(t)\} \label{vGxi<}
\leq{2C_Aw^{2}|t|^2},
\end{align}
where $C_A$ is defined in (\ref{CA}).
By (\ref{xifxit})  and the Schwarz inequality
\begin{align}
\mathbf{Var}\{\xi^{A}_n[\varphi]\}\leq\bigg(\int \mathbf{Var}^{1/2}\{\xi^{A}_n(t)\}|F[\varphi ](t)|dt%
\bigg)^2.  \label{var}
\end{align}
This, (\ref{condF}), and (\ref{vGxi<}) yield (\ref{VarF}).
\end{proof}

\medskip

In this chapter we find limiting covariance for $\xi^{A} _{n}[\varphi ]$
and prove that $\xi^{A} _{n}[\varphi ]$ in GOE case satisfies CLT. We have two theorems:

\begin{theorem}
\label{t:covGE} Let $\widehat{M}^{(n)}$ be the GOE matrix (\ref{GOE}), and $%
\varphi _{1,2}:\mathbb{R\rightarrow R}$ be  test functions satisfying (\ref{condF}). Denote
\begin{equation*}
\mathbf{Cov}\{\xi^{A} _{n}[\varphi _{1}],\xi^{A} _{n}[\varphi _{2}]\}=\mathbf{E}%
\{\xi _{n}^{A\circ }[\varphi _{1}]\xi^{A} _{n}[\varphi _{2}]\}.
\end{equation*}%
Then we have
\begin{align}
C_{GOE}[\varphi _{1},\varphi _{2}]:& =\lim_{n\rightarrow \infty }\mathbf{Cov}%
\{\xi^{A} _{n}[\varphi _{1}],\xi^{A} _{n}[\varphi _{2}]\}  \notag \\
& =\frac{T_{A}^{2}}{2\pi ^{2}}\int_{-2w}^{2w}\int_{-2w}^{2w}\frac{\Delta
\varphi _{1}}{\Delta \lambda }\frac{\Delta \varphi _{2}}{\Delta \lambda }%
\frac{4w^{2}-\lambda _{1}\lambda _{2}}{\sqrt{4w^{2}-\lambda _{1}^{2}}\sqrt{%
4w^{2}-\lambda _{2}^{2}}}d\lambda _{1}d\lambda _{2}  \notag \\
& \;\;\;+(T_{A(A+A^{T})}/2-T_{A}^{2})\int_{-2w}^{2w}\int_{-2w}^{2w}\Delta
\varphi _{1}\Delta \varphi _{2}\rho _{sc}(\lambda _{1})\rho _{sc}(\lambda
_{2})d\lambda _{1}d\lambda _{2},  \label{CGOE}
\end{align}%
where $T_{A}$ is defined in (\ref{TA}),
\begin{equation}
\Delta \varphi =\varphi (\lambda _{1})-\varphi (\lambda _{2}),\quad \Delta
\lambda =\lambda _{1}-\lambda _{2},  \label{deltaf}
\end{equation}%
and $\rho _{sc}$ is the density of the semicircle law (\ref{rhosc}).%
\end{theorem}

\begin{theorem}
\label{t:cltGE} Let $\widehat{M}^{(n)}$ be the GOE matrix (\ref{GOE}), and $%
\varphi :\mathbb{R\rightarrow R}$ satisfies (\ref{condF}). Then the random variable $\xi^{A\circ}_n[\varphi]$ converges in
distribution to the Gaussian random variable with zero mean and the variance
given by \begin{align}
V_{GOE}[\varphi]=&\frac{T_A^2}{2\pi^2}\int_{-2w}^{2w}\int_{-2w}^{2w}\Big(%
\frac{\Delta \varphi }{\Delta \lambda}\Big)^2 \frac{4w^2-\lambda_1\lambda_2}{%
\sqrt{4w^2-\lambda_1^2}\sqrt{4w^2-\lambda_2^2}} d\lambda_{1}d\lambda _{2}
\notag \\
& +(T_{A(A+A^T)}/2-T_{A}^{2})\int_{-2w}^{2w}\int_{-2w}^{2w}\big(\Delta \varphi%
\big)^2\rho _{sc}(\lambda _{1})\rho _{sc}(\lambda _{2})d\lambda _{1}d\lambda
_{2}.  \label{VarGOE}
\end{align}
\end{theorem}

\begin{remark}\label{r:GOE} Note that  $V_{GOE}[\varphi]$ can be written in the  form
\begin{align}
V_{GOE}[\varphi]=T_A^2\cdot V^{\mathcal{N}_n}_{GOE}[\varphi] +(T_{A(A+A^T)}/2-T_{A}^{2})\cdot V^{jj}_{GOE}[\varphi],  \label{VarGOE1}
\end{align}
where
\begin{align}
& V^{\mathcal{N}}_{GOE}[\varphi]=\frac{1}{2\pi^2}\int_{-2w}^{2w}\int_{-2w}^{2w}\Big(%
\frac{\Delta \varphi }{\Delta \lambda}\Big)^2 \frac{4w^2-\lambda_1\lambda_2}{%
\sqrt{4w^2-\lambda_1^2}\sqrt{4w^2-\lambda_2^2}} d\lambda_{1}d\lambda _{2} \label{VN}
\end{align}
and
\begin{align}
V^{jj}_{GOE}[\varphi]=\int_{-2w}^{2w}\int_{-2w}^{2w}\big(\Delta \varphi%
\big)^2\rho _{sc}(\lambda _{1})\rho _{sc}(\lambda _{2})d\lambda _{1}d\lambda
_{2}  \label{Vjj}
\end{align}
 are the limiting  variances corresponding to the  linear eigenvalue
statistics (\ref{Nn}) and  matrix
elements (\ref{fjj}), respectively (compare with the results of \cite{Ly-Pa:08} and  \cite{Ly-Pa:09}).

Besides, we have for limiting variance $V^{(M\eta,\eta)}_{GOE}[\varphi]$,
corresponding to the bilinear form (\ref{bf}):
 \begin{align}
V^{(M\eta,\eta)}_{GOE}[\varphi]=V^{jj}_{GOE}[\varphi]=\int_{-2w}^{2w}\int_{-2w}^{2w}\big(\Delta \varphi%
\big)^2\rho _{sc}(\lambda _{1})\rho _{sc}(\lambda _{2})d\lambda _{1}d\lambda
_{2} . \label{Veta}
\end{align}
\end{remark}
\begin{proof} {\bf Theorem \ref{t:covGE}.} Since $\mathbf{Cov}\{\xi^{A}_n[\varphi_{1}],\xi^{A}_n[\varphi_{2}]\}$ is linear in $\varphi _{1,2}$, it suffices to consider real valued $%
\varphi _{1,2}$. %
Writing the Fourier inversion formula (\ref{invFT})
and using the linearity of $\mathbf{Cov}\{\xi^{A}_n[\varphi_{1}],\xi^{A}_n[\varphi_{2}]\}$ in $\varphi _{1,2}$ and the spectral theorem
for symmetric  matrices, we obtain%
\begin{equation}\label{CF1F2}
\mathbf{Cov}\{\xi^{A}_n[\varphi_{1}],\xi^{A}_n[\varphi_{2}]\}
=\int \int \mathbf{Cov}\{\xi^{A}_n(t_{1}),\xi^{A}_n(t_{2})\}
F[\varphi_{1}](t_{1})F[\varphi_{2}](t_{2})dt_{1}dt_{2}
\end{equation}%
with $\xi^{A}_n(t_{})$ of (\ref{xint}). Similar to (\ref{vGxi<}) with the help of Poincar\'{e} inequality (\ref{Nash})
it can be shown that
\begin{align*}
&\mathbf{Var}\{\xi^{A\prime}_n(t)\}\leq
 ct^2,
\end{align*}
where $\xi^{A\prime}_n(t)=i\mathrm{Tr}A^{(n)}\widehat{M}e^{it_{1}\widehat{M}}$. This,  (\ref{vGxi<}),  and the Schwarz inequality imply the bounds
\begin{align}
&\left|\mathbf{Cov}\{\xi^{A}_n(t_{1}),\xi^{A}_n(t_{2})\}\right\vert
\leq c|t_{1}||t_{2}|,
\label{Covoc}
\\
&\left|\partial\ \mathbf{Cov}\{\xi^{A}_n(t_{1}),\xi^{A}_n(t_{2})\}/\partial t_{i}\right\vert
\leq c|t_{1}||t_{2}|,\quad i=1,2. \label{dCovoc}
\end{align}%
Hence, in view of (\ref{condF}) the integrand in (\ref{CF1F2}) admits an integrable and $n$%
-independent upper bound, and by dominated
convergence theorem it suffices to prove the pointwise in $t_{1,2}$ convergence of $\mathbf{Cov}\{\xi^{A}_n(t_{1}),\xi^{A}_n(t_{2})\}$ to a certain limit
as $n\rightarrow \infty $, implying (\ref{CGOE}). It also follows from (\ref{Covoc})
-- (\ref{dCovoc}) that there exists a convergent subsequence
 $\{\mathbf{Cov}\{\xi^{A}_{n_j}(t_{1}),\xi^{A}_{n_j}(t_{2})\}\}_{j=1}^\infty$. We will show that every such a subsequence
has the same limit leading through (\ref{CF1F2}) to (\ref{CGOE}).

We can confine
ourselves to $t_{1,2}\geq 0$, because $\mathbf{Cov}\{\xi^{A}_n(-t_{1}),\xi^{A}_n(t_{2})\}
=\overline{\mathbf{Cov}\{\xi^{A}_n(t_{1}),\xi^{A}_n(t_{2})\}}$.%

Consider
\begin{equation}
\mathbf{Cov}\{\xi^{A}_n(t_{1}),\xi^{B}_n(t_{2})\}=\mathbf{E}%
\{\xi^{A}_n(t_{1})\xi^{B\circ}_n(t_{2})\},
\label{CAB}
\end{equation}%
putting
in appropriate moment $A^{(n)}=B^{(n)}$. Here $\xi^{A,B}_n(t_{1})$
correspond to $A^{(n)}$, $B^{(n)}$ satisfying  (\ref{1}) -- (\ref{TA}) (see
 (\ref{xint})). By using  Duhamel formula (\ref{Duh}) we can write
\begin{eqnarray*}
\mathbf{Cov}\{\xi^{A}_n(t_{1}),\xi^{B}_n(t_{2})\} &=&\mathbf{E}%
\{\mathrm{Tr}A^{(n)}U(t_{1})\xi^{B\circ}_n(t_{2})\} \notag\\
&=&i\int_{0}^{t_{1}}\sum_{l,m=1}^{n}\mathbf{E}\{\widehat{M}_{lm}(A^{(n)}
U)_{lm}(t_{3})\xi^{B\circ}_n(t_{2})\}dt_{3}.
\end{eqnarray*}%
 Applying  differentiation formula (\ref{diffga}) with (\ref{Wmom12}) written in the form
\begin{equation}
 \mathbf{E}\{(W_{lm}^{(n)})^{2}\}=w^{2}\beta_{lm}^{-1}  \label{Wmom2}
\end{equation}
(see (\ref{beta})), and then (\ref{DAU}) --
(\ref{Dxi}), we obtain:
\begin{eqnarray}
\mathbf{Cov}\{\xi^{A}_n(t_{1}),\xi^{B}_n(t_{2})\} &=&
iw^{2}\int_{0}^{t_{1}}\frac{1}{n}\sum_{l,m=1}^{n}\beta_{lm}^{-1}\mathbf{E}\{D_{lm}[(A^{(n)}
U)_{lm}(t_{3})\xi^{B\circ}_n(t_{2})]\}dt_{3}\label{Cov1} \\
&=&
-\frac{w^{2}}{n}
\int_{0}^{t_{1}}dt_{3}\int_{0}^{t_{3}}\mathbf{E}%
\{[\xi^{I}_n(t_{3}-t_{4})\xi^{A}_n(t_{4})+\xi^{A}_n(t_{3})]\xi^{B\circ}_n(t_{2})\}dt_{4}   \notag\\
&&-\frac{w^{2}}{n}
\int_{0}^{t_{1}}dt_{3}\int_{0}^{t_{2}}\mathbf{E}%
\{\mathrm{Tr}A^{(n)}U(t_{3}+t_{4})(B^{(n)}+B^{(n)T})U(t_{2}-t_{4})\}dt_{4}.  \notag
\end{eqnarray}%
Putting
\begin{align}
&v_{n}=v^\circ_{n}+\overline{v}_{n},\;\overline{v}_{n}=\mathbf{E}%
\{v_{n }\}, \label{Evn} \\
&\xi^{A}_n=\xi^{A\circ}_n+\overline{\xi}^{A}_n,\;\overline{\xi}^{A}_n=\mathbf{E}%
\{\xi^{A}_n\},\label{Exin}
\end{align}%
we get from (\ref{Cov1})
\begin{align}
\mathbf{Cov}\{\xi^{A}_n(t_{1}),\xi^{B}_n(t_{2})\}
=&
-{w^{2}}
\int_{0}^{t_{1}}dt_{3}\int_{0}^{t_{3}}\overline{v}_{n}(t_{3}-t_{4})\mathbf{E}%
\{\xi^{A}_n(t_{4})\xi^{B\circ}_n(t_{2})\}dt_{4}  \label{Cov2} \\
&-\frac{w^{2}}{n}
\int_{0}^{t_{1}}dt_{3}\int_{0}^{t_{3}}\overline{\xi}^{A}_n(t_{3}-t_{4})\mathbf{E}%
\{\xi^{I}_n(t_{4})\xi^{B\circ}_n(t_{2})\}dt_{4} \notag\\
&-\frac{w^{2}}{n}
\int_{0}^{t_{1}}dt_{3}\int_{0}^{t_{2}}\mathbf{E}%
\{\mathrm{Tr}A^{(n)}U(t_{3}+t_{4})(B^{(n)}+B^{(n)T})U(t_{2}-t_{4})\}dt_{4} \notag
\\
&+r_{n}(t_1,t_2),  \notag
\end{align}
where
\begin{eqnarray}
r_{n}(t_1,t_2)
&=&-{w^{2}}
\int_{0}^{t_{1}}\mathbf{E}%
\{[(v^{\circ}_{n}*\xi^{A\circ}_n)(t_{3})+t_{3}n^{-1}\xi^{A}_n(t_{3})]\xi^{B\circ}_n(t_{2})\}dt_{3}.  \label{rn}
\end{eqnarray}
 With the help of Poincar\'{e} inequality (\ref{Nash})
it can be shown that
\begin{equation*}
\mathbf{Var}\left\{v^{\circ}_{n}\xi^{A\circ}_n\right\} =O(n^{-2}),\quad n\rightarrow\infty,
\end{equation*}%
 which together with  (\ref{vGxi<}) yield
\begin{equation}
r_n(t_1,t_2)=O(n^{-1}),\quad n\rightarrow\infty.
\label{rO}
\end{equation}
  Consider convergent subsequence  $\{\mathbf{Cov}\{\xi^{A}_{n_j}(t_{1}),\xi^{B}_{n_j}(t_{2})\}\}_{j=1}^\infty$
and denote
\begin{equation*}
C^{A,B}(t_1,t_2):=\lim_{n_{j}\rightarrow\infty}\mathbf{Cov}\{\xi^{A}_{n_j}(t_{1}),\xi^{B}_{n_j}(t_{2})\}. \end{equation*}
It follows from (\ref{vnvt}), (\ref{xiTA}), and (\ref{Cov2}) -- (\ref{rO})
that  $C^{A,B}(t_1,t_2)$ satisfies the equation
\begin{align}
C^{A,B}(t_1,t_2)+&
{w^{2}}
\int_{0}^{t_{1}}(v*C^{A,B}(\cdot,t_2))(t_3)dt_{3}  \label{Cov3} \\
=&-{w^{2}}T_A
\int_{0}^{t_{1}}(v*C^{I,B}(\cdot,t_2))(t_3)dt_{3} \notag\\
&-{w^{2}}
\int_{0}^{t_{1}}dt_{3}\int_{0}^{t_{2}}\lim_{n_{j}\rightarrow\infty}\mathbf{E}%
\{n_j^{-1}\mathrm{Tr}A^{(n)}U(t_{3}+t_{4})(B^{(n)}+B^{(n)T})U(t_{2}-t_{4})\}dt_{4}.  \notag
\end{align}
In particular, putting in (\ref{Cov3}) $A^{(n)}=I$ we get
\begin{equation}
C^{I,B}(t_1,t_2)
+
{2w^{2}}
\int_{0}^{t_{1}}(v*C^{I,B}(\cdot,t_2))(t_3)dt_{3} =  -2{w^{2}}T_Bt_2
\int_{0}^{t_{1}}v(t_{2}+t_{3})dt_{3}, \label{CovI}
\end{equation}
so that by (\ref{F3})
\begin{equation}
C^{I,B}(t_1,t_2)=\frac{T_B}{2\pi ^{2}}\int_{-2w}^{2w}\; \int_{-2w}^{2w}\frac{\Delta e(t_{1})\Delta e(t_{2})}{(\lambda_{1} -\lambda_{2} )^{2}}\frac{4w^2-\lambda_1\lambda_2}{%
\sqrt{4w^2-\lambda_1^2}\sqrt{4w^2-\lambda_2^2}} d\lambda_{1}d\lambda _{2}.\label{CI}
\end{equation}
\noindent
Now let us calculate the second term in the r.h.s. of (\ref{Cov3}).
Consider $\mathbf{E}%
\{n^{-1}\mathrm{Tr}A^{(n)}U(t_{1})C^{(n)}U(t_{2})\}$. We have by (\ref{EUjk})
\begin{eqnarray}
\mathbf{E}%
\{n^{-1}\mathrm{Tr}A^{(n)}U(t_{1})C^{(n)}U(t_{2})\}=T_{AC}\overline{v}_n(t_1)\overline{v}_{n}(t_2)+F_n(t_{1},t_{2}),
  \label{AUCU}
\end{eqnarray}
where $\overline{v}_n$ is defined in (\ref{Evn}) and
\begin{eqnarray*}
F_n(t_{1},t_{2})=n^{-1}\sum_{j,l=1}^n\mathbf{E}%
\{(UC^{(n)})_{jl}(t_{1})(U^{\circ}A^{(n)})_{lj}(t_{2})\}.
\end{eqnarray*}
Repeating steps leading from (\ref{CAB}) to (\ref{Cov3}) and using consequently  Duhamel formula
(\ref{Duh})
 and the
differentiation formulas (\ref{diffga}) and  (\ref{DAU}) --
(\ref{Dxi}), one can easily get
\begin{align}
F_n(t_{1},t_{2})
&+{w^{2}}
\int_{0}^{t_{1}}(\overline{v}_{n}*F_n(\cdot,t_{2}))(t_{3})dt_{3} \label{Fn}
\\
&=-{w^{2}}
\int_{0}^{t_{1}}dt_{3}\int_{0}^{t_{2}}n^{-1}\overline{\xi}^{A}_n(t_{3}+t_{4})\cdot
n^{-1}\overline{\xi}^{C}_n(t_{2}-t_{4})dt_{4}+r^{1}_{n}(t_1,t_2),  \notag
\end{align}
with $\overline{\xi}^{A}_n$   of  (\ref{Exin}), and
\begin{align*}
r^{1}_{n}(t_1,t_2)
=-{w^{2}}&
\int_{0}^{t_{1}}dt_{3}\mathbf{E}%
\Big\{\Big[(v^{\circ}_{n}*F _n(\cdot,t_{2}))(t_{3})
+t_{3}n^{-1}F_n(t_{3},t_{2})  \label{rn1}
\\
&+n^{-2}
\int_{0}^{t_{2}}\mathrm{Tr}A^{(n)T}U(t_{3}+t_{4})C^{(n)}U(t_{2}-t_{4})+\xi^{A\circ}_n(t_{3}+t_{4})\xi^{C\circ}_n(t_{2}-t_{4})\big)dt_{4}\Big]\Big\}.
\end{align*}
It follows from (\ref{vGxi<}) that
\begin{equation*}
r^{1}_n(t_1,t_2)=O(n^{-1}),\quad n\rightarrow\infty.
\end{equation*}
This,  (\ref{vnvt}), (\ref{xiTA}), and  (\ref{Fn}), yield
for $F=\lim_{n_j\rightarrow\infty}F_{n_j}$:
\begin{align}
F(t_{1},t_{2})
&+{w^{2}}
\int_{0}^{t_{1}}(v*F(\cdot,t_{2}))(t_{3})dt_{3} \notag
\\
&=-{w^{2}}T_AT_C
\int_{0}^{t_{1}}dt_{3}\int_{0}^{t_{2}}v(t_{3}+t_{4})v(t_{2}-t_{4})dt_{4}.
 \label{eF}
\end{align}
Hence, $F(t_{1},t_{2}) =T_AT_CF_1(t_1,t_2)$ with $F_1$ of (\ref{F1}). This,
(\ref{AUCU}), and (\ref{vnvt}) yield
\begin{align*}
F(t_{1},t_{2})& =
-{w^{2}}T_AT_C
\int_{0}^{t_{1}}dt_{3}\int_{0}^{t_{2}}v(t_{1}-t_{3})v(t_{3}+t_{4})v(t_{2}-t_{4})dt_{4}
 \notag
 \\
 &={}\frac{T_AT_C}{2}
\int_{-2w}^{2w}\int_{-2w}^{2w}\Delta e(t_1)\Delta e(t_2)\rho _{sc}(\lambda _{1})\rho _{sc}(\lambda
_{2})d\lambda _{1}d\lambda _{2}
\\
&=T_AT_C\big(v(t_1+t_2)-v(t_1)v(t_2)\big),
\end{align*}
where $\Delta e$ is defined in (\ref{e12}). This, (\ref{vt}), and (\ref{AUCU}) leads to
\begin{eqnarray}
\lim_{n_{j}\rightarrow\infty}\mathbf{E}%
\{n^{-1}\mathrm{Tr}A^{(n)}U(t_{1})C^{(n)}U(t_{2})\}=T_AT_Cv(t_1+t_2)+\big(T_{AC}-T_AT_C\big) v(t_1)v(t_2).
  \label{F}
\end{eqnarray}
Putting (\ref{CI}) and (\ref{F}) with $C^{(n)}=B^{(n)}+B^{(n)T}$ in (\ref{Cov3}) we obtain the equation
for $C^{A,B}$
\begin{align*}
C^{A,B}(t_1,t_2)
&+
{w^{2}}
\int_{0}^{t_{1}}(v*C^{A,B}(\cdot,t_2))(t_3)dt_{3}   \\
&=-{w^{2}}
\int_{0}^{t_{1}}\Big[T_A(v*C^{I,B}(\cdot,t_2))(t_3)+2T_AT_Bt_{2}v(t_3+t_2) \notag\\
&\quad\quad\quad\quad\big(T_{A(B+ B^T)}-2T_AT_B\big)
\int_{0}^{t_{2}}v(t_{3}+t_{4})v(t_{2}-t_{4})dt_{4}\Big]dt_{3},  \notag
\end{align*}
solving which with the help of  Lemma \ref{l:eq} we finally get
\begin{align*}
C^{A,B}(t_1,t_2)=&\frac{T_{A}T_B}{2\pi ^{2}}\int_{-2w}^{2w}\; \int_{-2w}^{2w}\frac{\Delta e(t_{1})\Delta e(t_{2})}{(\lambda_{1} -\lambda_{2} )^{2}}\frac{4w^2-\lambda_1\lambda_2}{%
\sqrt{4w^2-\lambda_1^2}\sqrt{4w^2-\lambda_2^2}} d\lambda_{1}d\lambda _{2}
\\
&+\big(T_{A(B+B^T)}/2-T_AT_B\big)
\int_{-2w}^{2w}\int_{-2w}^{2w}\Delta e(t_1)\Delta e(t_2)\rho _{sc}(\lambda _{1})\rho _{sc}(\lambda
_{2})d\lambda _{1}d\lambda _{2}.\notag
\end{align*}
 Putting this expression with $A^{(n)}=B^{(n)}$
 in (\ref{CF1F2})
  we obtain (\ref{CGOE})
and so prove the theorem.
%
%
\end{proof}

\medskip

\noindent \begin{proof} {\bf Theorem \ref{t:cltGE}.} The detailed proofs of CLTs for
linear eigenvalue statistics (\ref{Nn}) and for matrix elements (\ref{fjj})
are given in \cite{L-Pa:08,Ly-Pa:08} and \cite{Ly-Pa:09}, respectively. The
proof of   Theorem \ref{t:cltGE} follows the same scheme, so here we only
outline its main steps.

By the continuity theorem for characteristic functions it suffices to show
that if  \begin{equation}
Z^{A}_{n}(x)=\mathbf{E}\big\{e^{ix\xi_{n}^{A\circ }[\varphi ]}\big\}%
,  \label{Znx}
\end{equation}%
then for any $x\in \mathbb{R}$
\begin{equation}
\lim_{n\rightarrow\infty}Z^{A}_n(x)=Z^{A}(x), \label{limZn}
\end{equation}
where
\begin{equation}
Z^{A}(x)=\exp \big\{-x^{2}V_{GOE}[\varphi ]/2\big\}.  \label{chfGa}
\end{equation}%
We obtain (\ref{chfGa}), hence the theorem, for a
class of test functions satisfying condition
\begin{equation}
\int (1+|t|^{2})|F[\varphi ](t)|dt<\infty .  \label{condF2}
\end{equation}%
 (cf (\ref{condF})),  then  the theorem can be extended to the
class of  functions satisfying (\ref{condF})  by using a
standard approximation procedure.

Since $Z_{n}^{A}(0)=1$ and $%
Z_{n}^{A}(x)$ is continuous, we can write the relation
\begin{equation}
Z_{n}^{A}(x)=1+\int_{0}^{x}Z_{n}^{A}{}'(y)dy,\quad x\in \mathbb{R,}
\label{ZnZnpj}
\end{equation}%
showing that it suffices to prove that the sequence $\{Z_{n}^{A}{}'\}$
is uniformly bounded on any finite interval and that for any converging
subsequences $\{Z_{n_{i}}^{A}\}_{i\geq 1}$ and $\{Z_{n_i}^{A}{}'\}_{i\geq
1}$ there exists $Z^{A}(x)$, such that
\begin{equation}
\lim_{i\rightarrow \infty }Z^{A}_{n_{i}}(x)=Z^{A}(x),  \label{lz1}
\end{equation}
and
\begin{align}
\lim_{i\rightarrow \infty }Z_{n_i}^{A}{}'(x)=-xV_{GOE}[\varphi ]Z^{A}(x).  \label{lz2}
\end{align}
Indeed, if yes, then $Z^{A}(x)$ is a continuous function satisfying equation
\begin{equation}
Z^{A}(x)=1-V_{GOE}[\varphi ]\int_{0}^{x}yZ^{A}(y)dy,\quad x\in\mathbb{R,}  \label{eqzga}
\end{equation}%
which is uniquely soluble in the class of bounded continuous
functions, and its solution is evidently (\ref{chfGa}).

We denote
\begin{equation}
e_{n}(x)=\exp\{ix\xi_{n}^{A\circ }[\varphi ]\},\quad   \label{ejn}
\end{equation}%
and write according to (\ref{invFT}) and (\ref{Znx})%
\begin{equation}
Z_{n}^{A\prime }(x)=i\mathbf{E}\left\{ \xi_{n}^{A\circ }[\varphi
]e^{ix\xi_{n}^{A\circ }[\varphi ]}\right\} =i\int F[\varphi
](t)Y^A_{n}(x,t)dt, \label{dZY}
\end{equation}%
where%
\begin{equation}
Y^A_{n}(x,t)=\mathbf{E}\left\{ \xi_{n}^{A}(t)e^{\circ }_{n}(x)\right\},
\label{Ynx}
\end{equation}%
and $ \xi_{n}^{A}(t)$ is defined in (\ref{xint}). It follows from the Schwarz inequality and (\ref{vGxi<})
that\begin{equation*}
|Y^A_{n}(x,t)|\leq c|t|.
\end{equation*}%
This and (\ref{condF2}) yield that the sequence $Z_{n}^{A\prime }$
is uniformly bounded. Hence, there is a convergent subsequence $%
Z_{n_{i}}^{A}$, and by the dominated convergence theorem to find its
limit as $n\rightarrow\infty$ it suffices to find the pointwise limit of the
corresponding  subsequence $Y^A_{n_{i}}$.
It also can be shown with the help of  Poincar\'{e} inequality (\ref{Nash})
and (\ref{condF2}) that sequences $\{\partial Y^A_{n}/\partial x\}$ and $%
\{\partial Y^A_{n}/\partial t\}$ are uniformly bounded in $(t,x)\in K\subset%
\mathbb{R}^2_+$, $n\in\mathbf{N}$, for any bounded $K$, so that the sequence
$\{Y^A_{n}\}$ is equicontinuous on any finite set of $\mathbb{R}^2_+$, and
contains convergent subsequences.
Hence, for any
converging subsequence $\{Z_{n_{i}}^{A}\}$ (see (\ref{lz1})) there is a
converging subsequence $\{Y^A_{n'_{i}}\}$  and continuous function $Y^{A}$ (which
obviously depends on $\{Z_{n_{i}}^{A}\}$) such that
\begin{equation}
\lim_{n^{\prime }_i\rightarrow\infty}Y^{A}_{n^{\prime }_i}=Y^{A},\quad
\lim_{n^{\prime }_i\rightarrow\infty}Z^{A}_{n^{\prime }_i}=Z^{A}.  \label{YZ}
\end{equation}
We will show now that $Y^{A}$ satisfies certain integral equation leading
through (\ref{dZY}) to (\ref{eqzga}), hence, to (\ref{chfGa}).  This will
finish the proof of the theorem under condition (\ref{condF2}).

Applying consequently  the Duhamel formula (\ref{Duh}) and differentiation formula (\ref{diffga})
with (\ref{Wmom2}), we get %
\begin{align*}
Y^A_{n}(x,t)&=\frac{i}{\sqrt{n}}\int_{0}^{t}\sum_{j,k=1}^{n}\mathbf{E}%
\{\widehat{W}_{jk}(UA^{(n)})_{kj}(t_{1})e_{n}^{\circ }(x)\}dt_{1}
\\
&=\frac{iw^{2}}{n}\int_{0}^{t}\sum_{j,k=1}^{n}\beta^{-1} _{jk}\mathbf{E}%
\{D_{jk}\big((UA^{(n)})_{kj}(t_{1})e_{n}^{\circ }(x)\big)\}dt_{1},
\end{align*}%
where $D_{jk}=\partial /\partial M_{jk}$. It follows from (\ref{Dxi}) that
\begin{equation}
D_{jk}{e}_{n}(x)=-\beta _{jk}xe_{n}(x)\int
(U*C^{(n)}U)_{jk}(\theta)F[\varphi ](\theta)d\theta. \label{ParE}
\end{equation}%
This, (\ref{DAU}), (\ref{Evn}) -- (\ref{Exin}), and relation $e_{n}=e^\circ_{n}+Z^{A}_{n}$ yield
\begin{align*}
Y^A_{n}(x,t)=& -w^{2}\int_{0}^{t}
\big[\overline{v}_{n}*Y^A_{n}(x,\cdot)+\overline{\xi}_{n}*Y^I_{n}(x,\cdot)\big](t_{1}) dt_{1} \\
& -iw^{2}xZ^{A}_{n}(x)\int_{0}^{t}dt_{1}\int F[\varphi
](\theta)d\theta\int_0^{\theta}n^{-1}
\mathbf{E}\left\{\mathrm{Tr}A^{(n)}U(\theta-\theta_{1})C^{(n)}U(\theta_{1}+t_{1})\right\}
d\theta_{1}
\\
&+r_n(x,t),
\end{align*}%
where
\begin{align*}
r_n(x,t)=& -w^{2}\int_{0}^{t}
\mathbf{E}\left\{\big[(v^{\circ}_{n}*\xi^{A\circ}_{n})(t_{1})+t_{1}n^{-1}\xi_{n}^{A}(t_{1})\big]e_{n}^{\circ }(x)\right\} dt_{1} \\
& -iw^{2}x\int_{0}^{t}dt_{1}\int F[\varphi
](\theta)d\theta\int_0^{\theta}n^{-1}\mathbf{E}\left\{\mathrm{Tr}A^{(n)}U(\theta-\theta_{1})C^{(n)}U(\theta_{1}+t_{1})e_{n}^{\circ
}(x)\right\} d\theta_{1}.
\end{align*}
With the help of Poincar\'{e} inequality (\ref{Nash})
it can be shown that
\begin{equation*}
\mathbf{Var}\left\{\mathrm{n^{-1}Tr}A^{(n)}U(\theta-\theta_{1})C^{(n)}U(\theta_{1}+t_{1})\right\} \leq c (|\theta|^{2}+|\theta_{1}|^{2}+|t|^{2})n^{-2},
\end{equation*}%
 which together with  (\ref{vGxi<}) and (\ref{condF2}) yield
\begin{equation}
r_n(t_1,t_2)=O(n^{-1}),\quad n\rightarrow\infty.
\label{rn2}
\end{equation}%
This and (\ref{vnvt}) leads to  equation with respect to $Y^A=\lim_{n_j\rightarrow\infty} Y^A_{n_j}$:
\begin{align}
Y^{A}(x,t)+&w^{2}\int_{0}^{t}({v}_{}*Y^A_{}(x,\cdot))(t_1)dt_{1}\label{eqYA}
\\
=&-w^{2}T_A\int_{0}^{t}
({v}_{}*Y^I_{}(x,\cdot))(t_1)dt_{1}\notag
\\
&-iw^{2}xZ^{A}(x)\int_{0}^{t}dt_{1}\int F[\varphi
](\theta)d\theta\int_0^{\theta}
\lim_{n_{j}\rightarrow\infty}\mathbf{E}%
\{n^{-1}\mathrm{Tr}A^{(n)}U(\theta-\theta_{1})C^{(n)}U(\theta_{1}+t_{1})\}d\theta_{1},
\notag
\end{align}%
where  $Y^I$ is a solution of the equation
\begin{align}
Y^{I}(x,t)+2&w^{2}\int_{0}^{t}({v}_{}*Y^I_{}(x,\cdot))(t_1)dt_{1}\label{eqYI}
\\
&=-2iw^{2}xZ^{A}(x)T_A\int_{0}^{t}dt_{1}\int F[\varphi
](\theta){\theta} v(\theta+t_{1})d\theta. \notag
\end{align}%
Comparing pairs of equations (\ref{Cov3}) -- (\ref{CovI}) and (\ref{eqYA})
--
(\ref{eqYI}) one can see that
\begin{align}
Y^{A}(x,t)=-ixZ^{A}(x)\int C^{A,A}(t,\theta)F[\varphi
](\theta)d\theta,
 \notag
\end{align}%
where $C^{A,A}$ is given by  (\ref{CAB}) with $A=B$.
 This and (\ref{dZY}) yield
\begin{equation*}
Z^{A\prime }(x)=xZ^{A}(x)\int \int C^{A,A}(t,\theta)F[\varphi
](t)F[\varphi ](\theta)dtd\theta=xZ^{A}(x)V_{GOE}[\varphi ]
\label{dxZY}
\end{equation*}%
(see (\ref{CF1F2}), (\ref{CAB}), and (\ref{VarGOE})), and so leads to (\ref{eqzga}) and completes the proof of the theorem.
 \end{proof}

\section{ Covariance for $\xi^{A}_n[\varphi]$ in the Wigner case}\label{s:cov}

We show first that if $M^{(n)}$ is the Wigner matrix with uniformly bounded eighth
moments of its entries, and the test-function $\varphi$ is essentially of
class $\mathbf{C}^4$, then the variance of $\xi^{A}_n[\varphi]$ is of the order $%
O(1)$ as $n\rightarrow\infty$. We have

\begin{lemma}
\label{l:Var} Let $M^{(n)}=n^{-1/2}W^{(n)}$ be the real symmetric Wigner matrix (\ref{MW}%
) -- (\ref{Wmom12}). Assume that:

(i) the third moments of its entries do not depend on $j$, $k$, and $n$:%
\begin{equation}  \label{mu3}
\mu _{3}=\mathbf{E}\big\{ ( W_{jk}^{(n)}){}^{3}\big\};
\end{equation}

(ii) the eighth moments are uniformly bounded:%
\begin{equation}  \label{w8<}
w_{8}:=\sup_{n\in \mathbb{N}}\max_{1\leq j,k\leq n}\mathbf{E}\big\{ (
W_{jk}^{(n)}){}^{8}\big\} <\infty.
\end{equation}

\noindent Then for any test-function $\varphi :\mathbb{R\rightarrow C}$,
whose Fourier transform
(\ref{FT}) satisfies the condition
\begin{equation}
\int(1+|t|)^{4}|F[\varphi ](t)|dt<\infty,  \label{F4<}
\end{equation}%
we have the bound
\begin{align}
\mathbf{Var}\{\xi^{A}_n[\varphi]\}:&=\mathbf{E}\{|\xi^{A\circ}_n[\varphi]|^2\}
\notag \\
&\leq c\Big( \int (1+|t|)^{4}|F[\varphi ](t)|dt\Big) ^{2}.
\label{VarF}
\end{align}
\end{lemma}
The proof of (\ref{VarF}) follows from (\ref{var}), (\ref{F4<}), and bound \begin{equation}
\mathbf{Var}\{\xi^{A}_n(t)\} \leq c(1+|t|)^8  \label{vxi<}
\end{equation}
(see (\ref{xiv})).

\begin{theorem}
\label{t:Cov} Let $M^{(n)}=n^{-1/2}W^{(n)}$ be the real symmetric Wigner matrix (\ref{MW}%
) -- (\ref{Wmom12}), whose  third and fourth moments do not depend
on $j$, $k$, and $n$:%
\begin{equation}
\mu _{3}=\mathbf{E}\big\{(W_{jk}^{(n)}){}^{3}\big\},\quad \mu _{4}=\mathbf{E}%
\big\{(W_{jk}^{(n)}){}^{4}\big\},  \label{mu4}
\end{equation}%
and the eighth moments are uniformly bounded (see (\ref{w8<})). \noindent
Let  $ \{A^{(n)}\}_{n=1}^\infty$
satisfies (\ref{1}) -- (\ref{TA}), $C^{(n)}=A^{(n)}+A^{(n)T}$, and there exist
\begin{align}
& K^{(1)}_{A}=\lim_{n\rightarrow \infty
}n^{-3/2}\sum_{l,m=1}^{n}A^{(n)}_{ll}C^{(n)}_{lm}  \label{K1} \\
& K^{(2)}_{A}=T_{A}\lim_{n\rightarrow \infty }n^{-3/2}\sum_{l,m=1}^{n}C^{(n)}_{lm},
\label{K2}
\\
& K^{(3)}_{A}=\lim_{n\rightarrow \infty }n^{-1}%
\sum_{m=1}^{n}A^{(n)}_{mm}\Big(A^{(n)}_{mm}-n^{-1}\Tr A^{(n)}\Big).
\label{K3}
\end{align}

\noindent
Then we have for any   $\varphi _{1,2}:\mathbb{R\rightarrow R}$   satisfying (\ref{F4<}%
):
\begin{align}
\lim_{n\rightarrow \infty }\mathbf{Cov}\{\xi^{A} _{n}[\varphi _{1}], \xi^{A}
_{n}[\varphi _{2}]\}=&C_{GOE}[\varphi _{1},\varphi _{2}] +C_{\kappa _{3}}[\varphi _{1},\varphi _{2}]+C_{\kappa _{4}}[\varphi _{1},\varphi _{2}], \label{Coff}
 \end{align}%
where $C_{GOE}[\varphi _{1},\varphi _{2}]$ is defined in (\ref{CGOE}),
\begin{align}
&C_{\kappa _{3}}[\varphi _{1},\varphi _{2}]=\frac{\kappa _{3}}{w^{6}}\int_{-2w}^{2w}\int_{-2w}^{2w}\lambda _{1}
\Big(K_{A}^{(1)}\big(\lambda_{2}^2-w^2\big)+K_{A}^{(2)}\Big
(\frac{2w^{4}}{4w^{2}-\lambda_{2}^2}-\lambda_{2}^2\Big)\Big)\label{C3}
\\&\hspace{4.5cm}
\times(\varphi _{1}(\lambda _{1})\varphi _{2}(\lambda _{2})+\varphi _{1}(\lambda _{2})\varphi _{2}(\lambda _{1}))\prod_{j=1}^{2}\rho _{sc}(\lambda
_{j})d\lambda _{j},\notag
\\
&C_{\kappa _{4}}[\varphi _{1},\varphi _{2}]=\frac{\kappa _{4}}{w^{8}}\bigg[K^{(3)}_{A}\prod_{j=1}^{2}\int_{-2w}^{2w}%
\varphi _{j}(\lambda )(w^{2}-\lambda ^{2})\rho _{sc}(\lambda )d\lambda\label{C4}
\\
&\hspace{4.5cm}+\frac{T^{2}_A}{2\pi^2}\prod_{j=1}^{2}\int_{-2w}^{2w}%
\varphi _{j}(\lambda )\frac{2w^{2}-\lambda^2}{\sqrt{4w^{2}-\lambda^2}}d\lambda\bigg],\notag
\end{align}%
$\kappa _{3}=\mu _{3}$, and
\begin{equation}
\kappa _{4}=\mu _{4}-3w^{4}  \label{k4}
\end{equation}%
is the fourth cumulant of the off-diagonal entries (see (\ref{cums})).
\noindent In particular,
\begin{align}
V_W&[\varphi]:=\lim_{n\rightarrow\infty}\mathbf{Var}\{\xi^{A}_n[\varphi]\}
=V_{GOE}[\varphi]+C_{\kappa _{3}}[\varphi ,\varphi] +C_{\kappa _{4}}[\varphi,\varphi] \label{VW}
 \end{align}
with $V_{GOE}[\varphi]$ of (\ref{VarGOE}).
\end{theorem}
\begin{remark}
\label{r:Nnfjj} Note that for the limiting variances $V^{\mathcal{N}}_{W}[\varphi]$ and $V^{jj}_{W}[\varphi]$  of  linear eigenvalue
statistics (\ref{Nn}) and  matrix
elements (\ref{fjj})
we get, respectively:\begin{align}
& V^{\mathcal{N}}_{W}[\varphi]=V^{\mathcal{N}}_{GOE}[\varphi]+\frac{\kappa _{4}}{2\pi^2w^8}\Big|\int_{-2w}^{2w}%
\varphi _{}(\lambda )\frac{2w^{2}-\lambda^2}{\sqrt{4w^{2}-\lambda^2}}
d\lambda\Big|^2 \label{VWN}
\end{align}
and
\begin{align}
V^{jj}_{W}[\varphi]=V^{jj}_{GOE}[\varphi]+\frac{\kappa _{4}}{w^{8}}\Big|\int_{-2w}^{2w}%
\varphi _{}(\lambda )(w^{2}-\lambda ^{2})\rho _{sc}(\lambda )d\lambda \Big|^2,  \label{VWjj}
\end{align}
where $V^{\mathcal{N}}_{GOE}[\varphi]$ and $V^{jj}_{GOE}[\varphi]$ are defined in (\ref{VN}) -- (\ref{Vjj}). This coincides with the results of \cite{Ly-Pa:08} and  \cite{Ly-Pa:11}.
\end{remark}

\begin{remark}
\label{r:bf} In case of  bilinear forms (see (\ref{eta1}) -- (\ref{bf})) $T_A=0$ and coefficients
$K_A^{(j)}$, $j=1,2,3$ of (\ref{K1}) -- (\ref{K3}) take form
\begin{align*}
& K^{(1)}_{A}=2\lim_{n\rightarrow \infty
}n^{-1/2}\sum_{m=1}^{n}\eta^{(n)}_{m}\sum_{l=1}^{n}(\eta_{l}^{(n)})^3, \quad   \\
& K^{(2)}_{A}=2T_A\big(\lim_{n\rightarrow \infty
}n^{-1/2}\sum_{m=1}^{n}\eta^{(n)}_{m}\big)^2=0,
\\
& K^{(3)}_{A}=\lim_{n\rightarrow \infty }%
\sum_{m=1}^{n}(\eta^{(n)}_{m})^4.
\end{align*}
In particular, if   $\eta^{(n)}_m=O(n^{-1/2})$, $n\rightarrow \infty$ for all $m=1,...,n$, then $K_A^{(j)}=0$, $j=1,2,3$, and we get for the limiting variance:
\begin{align}
V^{(M\eta,\eta)}_{W}[\varphi]=V^{(M\eta,\eta)}_{GOE}[\varphi]=\int_{-2w}^{2w}\int_{-2w}^{2w}\big(\Delta \varphi%
\big)^2\rho _{sc}(\lambda _{1})\rho _{sc}(\lambda _{2})d\lambda _{1}d\lambda
_{2}    \label{VWeta}
\end{align}
(see (\ref{Veta})).
\end{remark}

\begin{remark}
\label{r:cov}
\textit{} We choose here the Wigner matrix so that its first
two moments matches the first two moments of the GOE matrix (see (\ref%
{Wmom12})). This fact allows to use known properties of GOE and lies at the
basis of interpolation procedure widely used in the proof of Lemma \ref%
{l:main} below. In fact this condition is pure technical one, and we can
replace condition (\ref{Wmom12}) with more general one and consider Wigner
matrix $\widetilde{M}=n^{-1/2}\widetilde{W}$, satisfying
\begin{align}
&\mathbf{E}\{\widetilde{W}^{(n)}_{jk}\}=0,\quad 1\leq j\leq k\leq n,
\label{mom12} \\
&\mathbf{E}\{(\widetilde{W}^{(n)}_{jk})^{2}\}=w^{2},\; j\neq k,\quad \mathbf{%
E}\{(\widetilde{W}^{(n)}_{jj})^{2}\}=w_{2}w^{2},\; w_2>0.  \notag
\end{align}%
In this case there arise additional terms in (\ref{Coff}) and (\ref{VW})
proportional to $w_2-2$. In particular, we have for the corresponding
limiting variance
\begin{align} \label{VW2}
V_{\widetilde{W}}^{w_2}[\varphi]=V_W[\varphi]
+(w_2-2)w^{-2}\Big(&K^{(3)}_{A}\Big( %
\int_{-2w}^{2w}\varphi(\mu)\mu\rho _{sc}(\mu)d\mu\Big)^2
\\
&+
T_A^2\Big(\frac{1}{2\pi} %
\int_{-2w}^{2w}\frac{\varphi(\mu)\mu}{\sqrt{4w^2-\mu^2}} d\mu\Big)^2\bigg),\notag \end{align}
where $V_W[\varphi]$ is given by (\ref{VW}).
\end{remark}
\medskip

\noindent \begin{proof} We write as in the GOE case
(see (\ref{CF1F2})):
\begin{equation}\label{CF1F2W}
\mathbf{Cov}\{\xi^{A}_n[\varphi_{1}],\xi^{A}_n[\varphi_{2}]\}
=\int \int \mathbf{Cov}\{\xi^{A}_n(t_{1}),\xi^{A}_n(t_{2})\}\prod_{j=1}^2F[\varphi_{j}](t_{j})dt_{j},
\end{equation}%
 and note that in view of (\ref{vxi<})
and (\ref{F4<}) the integrand  admits an integrable and $n$%
-independent upper bound. By dominated
convergence theorem it suffices to prove the pointwise in $t_{1,2}$ convergence of $\mathbf{Cov}\{\xi^{A}_n(t_{1}),\xi^{A}_n(t_{2})\}$ to a certain limit
as $n\rightarrow \infty $, implying (\ref{Coff}). To do this we use known result
for the GOE matrix (see Theorem \ref{t:covGE}) and an interpolating procedure proposed in \cite{KKP}.

Let $\widehat{M}^{(n)}=n^{-1/2}\widehat{W}^{(n)}$ be the GOE matrix (\ref{GOE})
independent of $M^{(n)}$, and
\begin{equation}
\widehat{U}(t)=\widehat{U}^{(n)}(t):=e^{it\widehat{M}^{(n)}},\quad \widehat{\xi}^{A}_n(t)=\Tr A^{(n)}\widehat{U}(t).  \label{UGOE}
\end{equation}
Consider
the "interpolating" random matrix \begin{equation}
M^{(n)}(s)=s^{1/2}M^{(n)}+(1-s)^{1/2}\widehat{M}^{(n)},\quad 0\leq s\leq 1,  \label{Ms}
\end{equation}%
viewed as defined on the product of the probability spaces of matrices $W^{(n)}$
and $\widehat{W}^{(n)}$ (cf (\ref{xint})). We denote again by $\mathbf{E}\{\dots \}$ the
corresponding expectation in the product space. Since $M^{(n)}(1)=M^{(n)}$, $M^{(n)}(0)=%
\widehat{M}^{(n)}$, then putting
\begin{equation}
U(t,s)=U^{(n)}(t,s):=e^{itM^{(n)}(s)},\quad \xi_{n}^{A}(t,s)=\Tr A^{(n)}U(t,s),  \label{Uts}
\end{equation}%
we can write
\begin{align}
C^\Delta_{n}(t_1,t_2):&=\mathbf{Cov}\{\xi^{A}_n(t_{1}),\xi^{A}_n(t_{2})\}-
\mathbf{Cov}\{\widehat{\xi}^{A}_n(t_{1}),\widehat{\xi}^{A}_n(t_{2})\}\label{CDelta}
\\
&=\int_{0}^{1}\frac{\partial }{\partial s}%
\mathbf{E}%
\{\xi^{A}_n(t_{1},s)\xi^{A\circ}_n(t_{2},s)\} ds=c^\Delta_{n}(t_1,t_2)+c^\Delta_{n}(t_2,t_1),  \notag \\
&  \notag
\end{align}%
where
\begin{align}\label{cdelta}
c^\Delta_{n}(t_1,t_2)&=
\int_{0}^{1}%
\mathbf{E}\Big\{\frac{\partial }{\partial s}(\xi^{A}_n(t _{1},s))\cdot\xi^{A\circ}_n(t_{2},s)\Big\} ds\\
&=\frac{i}{2}\int_{0}^{1}\Big(\frac{1}{\sqrt{ns}}\sum_{l,m=1}^n \mathbf{E}\Big\{W^{(n)}_{lm}\Phi _{lm}\Big\} -\frac{1}{\sqrt{n(1-s)}}\sum_{l,m=1}^n \mathbf{E}\Big\{\widehat{W}_{lm}\Phi _{lm}\Big\}\Big)  ds\notag
\end{align}%
and
\begin{align}\label{Phi}
\Phi _{lm}=\Phi _{lm}(t_1,t_2,s)=(U*A^{(n)}U)_{ml}(t_{1},s)\xi^{A\circ}_n(t_{2},s). \end{align}%
A simple algebra based on (\ref{norU}) -- (\ref{UAUOn}) allows to obtain
\begin{equation}
|D_{lm}^q\Phi_{lm}|\leq C_{q}(1+|t_{1}|+|t_{2}|)^{q+1}n^{3/2},  \label{P<}
\end{equation}
with $C_{q}$ depending only on $q\in\mathbb{N}$.
Besides, since
\begin{equation*}
\frac{\partial}{\partial W_{lm}^{(n)}} =\sqrt{\frac{s^{}}{n^{}}}D_{lm}(s),
\quad D _{lm}(s)=\frac{\partial}{\partial M^{(n)}_{lm}(s)},
\end{equation*}
then every derivative with respect to $W^{(n)}_{lm}$ gives the
factor $n^{-1/2}$. Hence, applying differentiation formula (\ref{difgen}) with $\zeta =W^{(n)}_{lm}$, $p=6$, and $%
\Phi=\Phi_{lm}$ to every term of the first sum and differentiation formula (\ref{diffga}) to every term of the second sum in the r.h.s. of (\ref{cdelta}), we obtain (see also (\ref{difinf})):
\begin{align}
c^\Delta_{n}(t_1,t_2)&=\frac{i}{2}\int_{0}^{1}\Big[\sum_{j=2}^6 s^{(j-1)/2}T_j^{(n)}+\varepsilon_6\Big]%
ds,  \label{c}
\end{align}
where
\begin{equation}
T_j^{(n)}=\frac{1}{j!n^{(j+1)/2}}\sum_{l,m=1}^{n}\kappa _{j+1,lm}\mathbf{E}%
\big\{D_{lm}^{j}\Phi_{lm}\big\},\quad j=2,...,6,  \label{Tp}
\end{equation}
and by (\ref{b3}) and (\ref{P<})
\begin{equation}
|\varepsilon_6|\leq\frac{C_6w_8}{n^{4}}\sum_{l,m=1}^n\sup_{M\in \mathcal{S}%
_{n}}|D_{lm}^{7}\Phi_{lm}|\leq {c(1+|t_{1}|+|t_{2}|)^8}{n^{-1/2}}.  \label{e2<}
\end{equation}
Now it follows from Lemma \ref{l:T2T3} below that
\begin{align}
&\int \int
\Big[\frac{i}{2}\int_{0}^{1}
s^{1/2}\lim_{n\rightarrow\infty}(T^{(n)}_{2}(t_1,t_2)+T^{(n)}_{2}(t_2,t_1))ds\Big]
\prod_{j=1}^2F[\varphi_{j}](t_{j})dt_{j}=C_{\kappa _{3}}[\varphi _{1},\varphi
_{2}],\label{limT2}
\\
&\int \int
\Big[\frac{i}{2}\int_{0}^{1}
s^{}\lim_{n\rightarrow\infty}(T^{(n)}_{3}(t_1,t_2)+T^{(n)}_{3}(t_2,t_1))ds\Big]
\prod_{j=1}^2F[\varphi_{j}](t_{j})dt_{j}=C_{\kappa _{4 }}[\varphi _{1},\varphi _{2}],
\label{limT3}
 \end{align}
 and
 \begin{align}
&\lim_{n\rightarrow\infty}T^{(n)}_{j}=0,\quad j=4,5,6,\label{limTj}
 \end{align}
 with $C_{\kappa _{3}}[\varphi _{1},\varphi _{2}]$, $C_{\kappa _{4}}[\varphi _{1},\varphi _{2}]$ of (\ref{C3}) -- (\ref{C4}). This leads through (\ref{c}), (\ref{CDelta}), and (\ref{CGOE}) to (\ref{Coff}) -- (\ref{C4}) and completes the proof.
\end{proof}
\begin{lemma}\label{l:T2T3} Under conditions of  Theorem \ref{t:Cov} the
statements (\ref{limT2})  --  (\ref{limTj}) are valid.
\end{lemma}
\begin{proof}
   Consider  $T^{(n)}_{2}$ of (\ref{Tp}) and note that by (\ref{cums}) and (\ref{mu3})
   $\kappa_{3,lm}=\mu_3=\kappa_{3}$, and we have
\begin{align}
T^{(n)}_{2}(t_{1},t_{2},s)&=\frac{\kappa_{3}}{2n^{3/2}}\sum_{l,m=1}^n \mathbf{E}%
\{D^2_{lm}((U*A^{(n)}U)_{ml}(t_{1},s)\xi^{A\circ}_n(t_{2},s))\}  \label{lT2=}
\\
&=\frac{\kappa_{3}}{2n^{3/2}}\sum_{l,m=1}^n \mathbf{E}\{\xi^{A\circ}_n(t_{2},s)D^2_{lm}(U*A^{(n)}U)_{ml}(t_{1},s)  \notag
\\
&\hspace{2cm}+2D_{lm}(U*A^{(n)}U)_{ml}(t_{1},s)D_{lm}\xi^{A}_n(t_{2},s)  \notag
\\
&\hspace{2cm}+(U*A^{(n)}U)_{ml}(t_{1},s)D_{lm}^{2}\xi^{A}_n(t_{2},s)\}=:%
\kappa_{3}[T_{21}^{(n)}+T_{22}^{(n)}+T_{23}^{(n)}].  \notag
\end{align}
Consider $T_{21}^{(n)}$. It follows from (\ref{ParU}) and (\ref{DUAU}) that $D^2_{lm}(U*A^{(n)}U)_{ml}$ of $T^{(n)}_{21}$ gives the terms of the form
\begin{align}
&T_{21}^{1(n)}=n^{-3/2}\sum_{l,m=1}^nU_{lm}U_{lm}(UA^{(n)}U)_{lm},\label{T211}
\\
&T_{21}^{2(n)}=n^{-3/2}\sum_{l,m=1}^nU_{lm}U_{ll}(UA^{(n)}U)_{mm},\label{T212}
\\
&T_{21}^{3(n)}=n^{-3/2}\sum_{l,m=1}^nU_{ll}U_{mm}(UA^{(n)}U)_{lm},\label{T213}
\end{align}
Here for shortness we omit the sign of conjugation $"*" $  and arguments  of $U$. Besides, we   replace  $\beta_{lm}$ with $1$, that in
view of (\ref{UAUmm}) gives error terms
of the order $O(n^{-1/2})$, $n\rightarrow\infty$.
 It follows from the Schwarz inequality, (\ref{norU}),  and (\ref{UAUOn})
 that
 \begin{equation*}
T_{21}^{1(n)}=O(n^{-1/2}),\quad n\rightarrow\infty,\label{T211}
\end{equation*}
and from (\ref{norU}), (\ref{UAUmm})
 that
 \begin{align*}
T_{21}^{2(n)}&\leq n^{-3/2}||U||\cdot ||(U_{11},...,U_{nn})^T||
\cdot ||((UA^{(n)}U)_{11},...,(UA^{(n)}U)_{nn})^T||\notag
\\
&=O(n^{-1/2}),\quad n\rightarrow\infty.
\end{align*}
This and (\ref{vxi<}) yield
 \begin{equation}
\big|\mathbf{E}\{(T_{21}^{1(n)}+T_{21}^{2(n)})\xi^{A\circ}_n\}\big|
\leq cn^{-1/2}\mathbf{Var}\{\xi^{A}_n\}^{1/2}=O(n^{-1/2}),
\quad n\rightarrow\infty.\label{T2122}
\end{equation}
We also have
 \begin{equation}
T_{21}^{3(n)}=O(1),\quad n\rightarrow\infty.\label{T311}
\end{equation}
Let us show that
 \begin{equation}
\mathbf{E}\{T_{21}^{3(n)}\xi^{A\circ}_n\}=O(n^{-1/2}),\quad
n\rightarrow\infty.\label{T213}
\end{equation}
For this purpose consider
\begin{align*}
R_{n}=n^{-3/2}\sum_{l,m=1}^n\mathbf{E}\{U_{ll}(t_{1})
U_{mm}(t_{2})(UA^{(n)}U)_{lm}\xi^{A\circ}_n\}.
\end{align*}
Putting here $U_{jj}=\mathbf{E}\{U_{jj}\}+U_{jj}^\circ$, $j=l,m$, and using (\ref{xiv})
 we get
\begin{align}
R_{n}=&v(t_{1})v(t_{2})n^{-3/2}\sum_{l,m=1}^n\mathbf{E}\{(UA^{(n)}U)_{lm}
\xi^{A\circ}_n\}\label{R}
\\
&+v(t_{1})n^{-3/2}\sum_{l,m=1}^n\mathbf{E}\{
U^\circ_{mm}(t_{2})(UA^{(n)}U)_{lm}\xi^{A\circ}_n\}\notag
\\
&+n^{-3/2}\sum_{l,m=1}^n\mathbf{E}\{U^\circ_{ll}(t_{1})
U_{mm}(t_{2})(UA^{(n)}U)_{lm}\xi^{A\circ}_n\}+o(1),\quad n\rightarrow\infty.\notag
\end{align}
It follows from the Schwarz inequality,  (\ref{xiv}) and (\ref{Vareta})
that the first term in the r.h.s. of (\ref{R}) is of the order $O(n^{-1/2})$, $n\rightarrow\infty$.
We also have
in view of (\ref{1}) and (\ref{norU})
\begin{align*}
n^{-3/2}\Big|\sum_{l,m=1}^n
U^\circ_{mm}(t_{2})(UA^{(n)}U)_{lm}\Big|&\leq n^{-1} ||UAU||\cdot ||(U^\circ_{11},...,U^\circ_{nn})^T||
\\
&\leq n^{-1/2}
\big(\sum_{m=1}^n
|U^\circ_{mm}(t_{2})|^{2}\big)^{1/2}.
\end{align*}
  Hence, by the Schwarz inequality and (\ref{xiv})
\begin{align*}
\Big|n^{-3/2}\sum_{l,m=1}^n&\mathbf{E}\{
U^\circ_{mm}(t_{2})(UA^{(n)}U)_{lm}\xi^{A\circ}_n\}\Big|
\\
&\leq n^{-1/2}\Big(\sum_{m=1}^n\mathbf{Var}\{
U_{mm}(t_{2})\}\Big)^{1/2}\mathbf{Var}\{\xi^{A}_n\}^{1/2}=O(n^{-1/2}),\; n\rightarrow\infty.
\end{align*}
Thus, the second and third terms in the r.h.s. of (\ref{R}) are of the order $O(n^{-1/2})$, $n\rightarrow\infty$, and we get (\ref{T213}).
Now (\ref{T2122}) -- (\ref{T213}) yield for $T^{ (n)}_{21}$ of (\ref{lT2=}):
\begin{equation}
T_{21}^{(n)}=O(n^{-1/2}),\quad
n\rightarrow\infty.\label{T21}
\end{equation}
Applying (\ref{Dxi}) -- (\ref{DUAUlm}) to calculate $T_{22}^{(n)}$ and $T_{23}^{(n)}$
of  (\ref{lT2=}) we get terms of the form
\begin{align}
&n^{-3/2}\sum_{l,m=1}^nU_{lm}(UA^{(n)}U)_{lm}(UA^{(n)}U)_{lm},\notag
\\
&n^{-3/2}\sum_{l,m=1}^nU_{ll}(UA^{(n)}U)_{mm}(UA^{(n)}U)_{lm},
\label{O1}
\end{align}
where as it follows from the Schwarz inequality and (\ref{UAUlm}) -- (\ref{UAUmm})
the first term is of the order $O(n^{-1/2})$, and the second is of the order
$O(1)$, $n\rightarrow\infty$. Hence, we
are left with
\begin{align}
T_{22}^{(n)}+T_{23}^{(n)}
=-\frac{1}{n^{3/2}}\sum_{l,m=1}^n \mathbf{E}
\{&2(U_{ll}*(U*A^{(n)}U))_{mm}(t_{1})(U*C^{(n)}U)_{lm}(t_{2})\label{T2223}
\\
&+((U*C^{(n)}U))_{lm}(t_{1})(U_{mm}*(U*C^{(n)}U)_{ll})(t_{2})/2\}\notag
\\
&\quad\quad+O(n^{-1/2}),\quad
n\rightarrow\infty.  \notag
\end{align}
Now it follows from (\ref{lT2=}), (\ref{T21}), (\ref{T2223}), and (\ref{limw})  that
\begin{align}
\lim_{n\rightarrow\infty}T_{2}^{(n)}
(t_1,t_2)=\kappa_{3}\lim_{n\rightarrow\infty}(T_{22}^{(n)}+T_{23}^{(n)})(t_1,t_2)
=-\kappa_{3}\big[2T_2(t_1,t_2)+T_2(t_2,t_1)\big],\label{T2}
\end{align}
where
\begin{align}
T_2(t_1,t_2)=\big[(K_A^{(1)}-K_A^{(2)})(v*v*v)(t_1)+K_A^{(2)}(v*tv)(t_1)\big]
\cdot (v* v)(t_2)\label{T2t1t2}
\end{align}
with $K_A^{(1)}$, $K_A^{(2)}$ of (\ref{K1}) -- (\ref{K2}) and $v$ of (\ref{vt}).
We also have
 \begin{align}
&
(v*v)(t)=-{iw^{-2}}\int_{-2w}^{2w}e^{i\mu t}\mu\rho_{sc}(\mu)d%
\mu,  \label{vv}
\\
&
(v*tv)(t)={w^{-2}}\int_{-2w}^{2w}e^{i\mu t}\Big[1-\frac{2w^2}{4w^2-\mu^2}\Big]\rho_{sc}(\mu)d%
\mu,  \label{vtv}
\\
&(v\ast v\ast v)(t)=w^{-4}\int_{-2w}^{2w}e^{it\mu}(w^2-\mu^2) \rho
_{sc}(\mu)d\mu.\label{vvv}
\end{align}
Putting (\ref{vv}) -- (\ref{vvv}) in (\ref{T2t1t2}), then plugging the result
in the l.h.s. of
(\ref{limT2}) we get after some calculations  (\ref{limT2}).

Consider now $T_{3}^{(n)}$ of (\ref{Tp}):
\begin{align*}
T_3^{(n)}=\frac{1}{6n^{2}}\sum_{l,m=1}^{n}\kappa _{4,lm}\mathbf{E}%
\big\{D_{lm}^{3}\big((U*A^{(n)}U)_{ml}(t_{1},s)\xi^{A\circ}_n(t_{2},s)\big)\big\},\quad   \label{T3}
\end{align*}
where in view of (\ref{cums}) and (\ref{k4})
\begin{equation*}
\kappa _{4,lm}=\kappa _{4}-9\delta _{lm}w^{4}.
\end{equation*}
It follows from (\ref{UAUmm}) and (\ref{vxi<}) that in (\ref{T3}) we can replace $\kappa _{4,lm}$  with $\kappa _{4} $, which gives  error terms
of the order $O(n^{-1/2})$, $n\rightarrow\infty$, and write
\begin{align}
T_3^{(n)}&=\frac{\kappa _{4}}{6n^{2}}\sum_{l,m=1}^{n}\mathbf{E}%
\big\{\xi^{A\circ}_n\cdot D_{lm}^{3}(U*A^{(n)}U)_{ml}+3D_{lm}^{}\xi^{A}_n\cdot D_{lm}^{2}(U*A^{(n)}U)_{ml}\notag
\\
&\hspace{3cm}+3D_{lm}^{2}\xi^{A}_n\cdot D_{lm}^{}(U*A^{(n)}U)_{ml}+(U*A^{(n)}U)_{ml}\cdot D_{lm}^{3}\xi^{A}_n\big\},\notag   \label{T3}
\\
&=:\kappa _{4}[T_{31}^{(n)}+ T_{32}^{(n)}+T_{33}^{(n)}+T_{34}^{(n)}]+O(n^{-1/2}),\quad
n\rightarrow\infty.
\end{align}
 Treating $T^{(n)}_{31}$ similar to $T_{21}^{(n)}$ of (\ref{lT2=}) (see (\ref{T213}) --
 (\ref{T21})) one can get
\begin{equation}
T_{31}^{(n)}=O(n^{-1/2}),\quad
n\rightarrow\infty.\label{T31}
\end{equation}
Besides, it can be shown with the help of  (\ref{UAUOn}) --  (\ref{UAUmm}) and  (\ref{vxi<}) that all terms containing off-diagonal
entries $U_{lm}$ or $(UA^{(n)}U)_{lm}$ vanish in the limit
$n\rightarrow\infty$, hence,
\begin{align*}
 T_{32}^{(n)}+T_{34}=O(n^{-1/2}),
\end{align*}
and we are left with
\begin{align*}
&T_{3}^{(n)}=-\frac{i\kappa _{4}}{n^{2}}\sum_{l,m=1}^{n}\mathbf{E}%
\big\{\big(U_{ll}*(U*A^{(n)}U)_{mm}\big)(t_{1})
\\
&\hspace{3cm}\times
\big(U_{ll}*(U*C^{(n)}U)_{mm}+U_{mm}*(U*C^{(n)}U)_{ll}\big)(t_{2})\big)
\big\}+O(n^{-1/2}),\notag
\end{align*}
as $n\rightarrow\infty$. Now it follows from  (\ref{limvC})  that
\begin{align*}
\lim_{n\rightarrow\infty}T_{3}^{(n)}=-2i\kappa _{4}\big[K_A^{(3)}(v*v*v)(t_1)(v*v*v)(t_2)
+2T_A^2(v*tv)(t_1)(v*tv)(t_2)\big].
\end{align*}
This and (\ref{vtv}) -- (\ref{vvv}) yield after some calculations  (\ref{limT3}).

It remains to show (\ref{limTj}). It is much simpler because in this case
we have additional factors $n^{-1/2}$ (see (\ref{Tp})), so that treating $T_j$, $j=4,5,6$ similar to  $T_j$, $j=2,3$ one can easily get (\ref{limTj}).
This completes the proof of the lemma.
 \end{proof}

\section{ Limiting probability law  for $\xi^{A}_n[\varphi]$ }\label{s:main}

\begin{theorem}
\label{t:clt} Consider the real symmetric Wigner random matrix of the form%
\begin{equation}
M^{(n)}=n^{-1/2}W^{(n)},\quad W^{(n)}=\{W^{}_{jk}\in \mathbb{R},\;W^{}_{jk}=W^{}_{kj}=(1+%
\delta _{jk})^{1/2}V^{}_{jk}\}_{j,k=1}^{n},  \label{MW1}
\end{equation}%
where $\{V_{jk}\}_{1\leq j\leq k<\infty }$ are i.i.d. random variables such
that
\begin{equation*}
\mathbf{E}\{V_{11}\}=0,\quad \mathbf{E}\{V_{11}^{2}\}=w^{2},\quad
\end{equation*}%
and  functions $\ln \mathbf{E}\{e^{itV_{11}}\}$ and   $\mathbf{E}\{e^{it|V_{11}| }\}$ are entire.

Let  $ \{A^{(n)}\}_{n=1}^\infty$
satisfies (\ref{1}) -- (\ref{TA}), $C^{(n)}=A^{(n)}+A^{(n)T}$, and there exist
\begin{align}
&A_p=\lim_{n\rightarrow\infty}n^{-p/2} \Big(%
\sum_{l,m=1}^n(C_{lm}^{(n)})^p+(2^{(2 -p)/2}-1)\sum_{m=1}^n(C_{mm}^{(n)})^p\Big)/2,\quad
p\geq 3.
\label{Ap}
\end{align}
 Then for any $\varphi :\mathbb{R\rightarrow R}$, whose Fourier transform (\ref%
{FT}) satisfies (\ref{F4<}), the random variable $\xi^{A\circ}_n[\varphi]$ converges
in distribution as $n\rightarrow \infty $ to the random variable $\xi^{A}[\varphi] $ such
that
\begin{equation}  \label{ln}
\ln \mathbf{E}\{e^{ix\xi^{A}[\varphi] }\}=-{x^{2}}V_{W}[\varphi ]
/2+\sum_{p=3}^\infty\frac{\kappa_p A_p }{p!}(ix^*)^p,
\end{equation}%
where
\begin{align}
x^*=\frac{x}{w^2}\int_{-2w}^{2w}\varphi(\mu)\mu\rho_{sc}(\mu)d\mu,
\label{x}
\end{align}
$\rho_{sc}$ is the density of the semicircle law (\ref{rhosc}), and $%
V_W[\varphi]$ is given by (\ref{VW}).
\end{theorem}

\begin{remark}\label{r:main}
It can be shown that in the case of matrix $\widetilde{M}^{(n)}%
=n^{-1/2}V^{(n)}$, the Theorem \ref{t:clt} holds true  with
\begin{equation*}
\ln \mathbf{E}\{e^{ix\xi^{A}[\varphi] }\}=-V_{\widetilde{W}}^{1}[\varphi]x^{2}
/2+\sum_{p=3}^\infty\frac{\kappa_p \widetilde{A}_p }{p!}(ix^*)^p,
\end{equation*}%
where $V_{\widetilde{W}}^{1}[\varphi]$ is given by (\ref{VW2}) with $w_2=1$,
and
\begin{align*}
\widetilde{A}_p=\lim_{n\rightarrow\infty}n^{-p/2} \Big(%
\sum_{l,m=1}^n(C_{lm}^{(n)})^p+(2^{( -2p+1)/2}-1)\sum_{m=1}^n(C_{mm}^{(n)})^p\Big)/2.
\end{align*}
\end{remark}
\begin{remark}
\label{r:bfclt} In the case of matrix elements (see (\ref{fjj}))  $A_p=2^{p/2}$, and
 we obtain the result of \cite{Ly-Pa:11} (see Theorem 3.4).

In the case of  bilinear forms (see (\ref{eta1}) -- (\ref{bf}))
we have
for $A_p$ of (\ref{Ap}):
\begin{align*}
A_p=\lim_{n\rightarrow\infty}\Big(%
\big(\sum_{l=1}^n(\eta^{(n)}_{l})^p\big)^2+(2^{(2 -p)/2}-1)\sum_{l=1}^n(\eta^{(n)}_{l})^{2p}\Big),\quad
p\geq 3.
\end{align*}
In particular, if   $\eta^{(n)}_m=O(n^{-1/2})$, $n\rightarrow \infty$ for all $m=1,...,n$, then $A_p=0$, $p\geq 3$, and the
random variable $(\varphi(M^{(n)})^\circ\eta^{(n)},\eta^{(n)})$ converges in
distribution to the Gaussian random variable with zero mean and the variance $V^{(M\eta,\eta)}_{GOE}[\varphi]$
of (\ref{Veta}).
\end{remark}
\begin{remark}
\label{r:even}
It follows from Theorem \ref{t:clt} that if $\varphi$ is even, then the
random variable $\xi^{A\circ}_n[\varphi]$ converges in
distribution to the Gaussian random variable with zero mean and the variance
$V_{GOE}[\varphi]+C_{\kappa _{4}}[\varphi,\varphi]$ (see (\ref{C3}) -- (\ref{VW})).
\end{remark}
\noindent\begin{proof}
Note  that in view of (\ref{beta}) and (\ref{MW1}) we can write
\begin{equation}
W_{lm}=\beta_{lm}^{-1/2}V_{lm}.  \label{Wjk}
\end{equation}
Besides, since  $\ln\mathbf{E}\{e^{i tV_{11}}\}$ is entire then we have
\begin{align}
\sum_{p=1}^\infty\frac{x^p|\kappa_{p+1}|}{p!}<\infty,\quad \forall x>0, \quad
\label{ser<}
\end{align}
 where $\kappa_{p}
 $ is the $p$th cumulant of $V_{11}$.
%

Consider the characteristic functions
\begin{equation}
Z^A_{n}(x)=\mathbf{E}\left\{ e^{ix(\xi^{A}_n[\varphi])^{\circ}}\right\}
\label{Zn}
\end{equation}%
and
\begin{equation}
\widehat{Z}^A_{n}(x)
=\mathbf{E}\left\{ e^{ix(\widehat{\xi}^{A}_n[\varphi])^{\circ}}\right\},
\label{Znxj}
\end{equation}
where $\widehat{\xi}^{A}_n[\varphi]$ corresponds to the
GOE matrix $\widehat{M}^{(n)}=n^{-1/2}\widehat{W}^{(n)}$  (\ref{GOE}).
In view of Theorem \ref{t:cltGE},  (\ref{VW}), and (\ref{ln}) it suffices to  prove that for any $x\in \mathbb{R}$
\begin{equation}
\lim_{n\rightarrow \infty }\ln Z^A_{n}(x)/\widehat{Z}^A_{n}(x)=-(C_{\kappa _{3}}[\varphi ,\varphi] +C_{\kappa _{4}}[\varphi,\varphi])x^{2}/2+
\sum_{p=3}^\infty\frac{\kappa_p A_p }{p!}(ix^*)^p.
\label{lnGW}
\end{equation}
Following the idea of the proof of  Theorem \ref{t:covGE}  we introduce
the "interpolating" random matrix $M^{(n)}(s)$ (see (\ref{Ms})), put
\begin{align}
&Z^A_{n}(x,s)=\mathbf{E}\left\{ e_{n}(x,s)\right\},\quad
e_{n}(x,s)=e^{ix(\xi^{A,s}_n[\varphi])^{\circ}},\quad\label{Zns}
\\
&\xi^{A,s}_n[\varphi]=\Tr \varphi(M^{(n)}(s))A^{(n)}, \notag
\\
&\xi^{A,s}_n(t)=\Tr U(t,s)A^{(n)},\quad U(t,s)=e^{itM(s)},
\label{xins}
\end{align}%
 and write
\begin{align}
\ln Z^A_{n}(x)/\widehat{Z}^A_{n}(x)&=
\int_{0}^{1}\frac{\partial }{\partial s}%
\ln Z^A_{n}(x,s) ds  \label{lnZZ}
\\
&=-\frac{x}{2}\int_{0}^{1}\frac{ds}{Z^A_{n}(x,s)}\int
\Big(\frac{1}{\sqrt{ns}}\sum_{l,m=1}^n \mathbf{E}\Big\{W^{(n)}_{lm}\Psi _{lm}\Big\}\notag
\\
&\hspace{2cm} -\frac{1}{\sqrt{n(1-s)}}\sum_{l,m=1}^n \mathbf{E}\Big\{\widehat{W}_{lm}\Psi _{lm}\Big\}\Big) F[\varphi](t)dt,     \notag
\end{align}%
where
\begin{align}\label{Psi}
\Psi _{lm}=\Psi _{lm}(t,x,s)=(U*A^{(n)}U)_{ml}(t,s)e^{\circ}_{n}(x,s).
\end{align}%
(cf (\ref{CDelta}) -- (\ref{Phi})). Let us note that unlike
functions $\Phi_{lm}$ of (\ref{Phi}), having all derivatives
$D^p_{lm}\Phi_{lm}$ of the order $O(n^{3/2})$ (see (\ref{P<})), here
we have $D^p_{lm}\Psi_{lm}=O(n^{(p+1)/2})$, and
there is no such finite $p\in \mathbb{N}$ that $\varepsilon_p$ of (\ref%
{difgen}) vanishes as $n\rightarrow\infty$. Hence, instead of (\ref%
{difgen}),  used while treating  (\ref{CDelta}), here for every term of the
first sum of the r.h.s. of (\ref{lnZZ}) we apply infinite version of (\ref%
{difgen})
 given by (\ref{difinf}) (see also (\ref{inf})). To
do this we check first that $\Psi_{lm}(x,t)$ satisfies condition
(\ref{al}). Assume  that the Fourier transform (\ref{FT}) of $\varphi $ satisfies
\begin{equation}
\int |F[\varphi ](t)||t|^{l}dt<C_{\varphi}l! \quad\forall l\in
\mathbb{N} ,  \label{phil<}
\end{equation}%
where $C_{\varphi}$ is an absolute constant. Using the Leibnitz rule %
we obtain
\begin{align}
D_{lm}^p\Psi_{lm}(x,t,s)=\sum_{q=0}^{p} \Big(%
\begin{array}{ll}
p &  \\
q &
\end{array}%
\Big)
D_{lm}^{p-q}(U*A^{(n)}U)_{ml}(t,s)D_{lm}^{q}e^{\circ}_{n}(x,s),
\label{DjkF=}
\end{align}
where
\begin{equation}
D_{lm}^{q}e_{n}(x,s)=ixD_{lm}^{q-1}\big(e_{n}(x,s)D_{lm}\xi^{A,s}_n[\varphi]%
\big),  \label{De=}
\end{equation}
(see (\ref{Zns})), so that
\begin{align*}
&D_{lm}^{q}e_{n}(x,s)=e_{n}(x,s)\sum_{r=1}^q(ix)^r\sum_ {%
\begin{array}{ll}
\overline{q}=(q_1,...,q_r): &  \\
q_1+...+q_r=q &
\end{array}%
} C_{\overline{q},r}\prod_{t=1}^rD_{lm}^{q_t}\xi^{A,s}_n[\varphi],
\end{align*}
and
\begin{equation*}
\sum_{\overline{q},r}C_{\overline{q},r}\leq 2^q.
\end{equation*}
Hence,
\begin{equation*}
|D_{lm}^{q}e_{n}(x,s)|\leq\big(2(1+|x|)\big)^q \max_ {1\leq r\leq
q,\;\sum_{t=1}^r q_t=q}
\prod_{t=1}^r|D_{lm}^{q_t}\xi^{A,s}_n[\varphi]|,
\end{equation*}
where
\begin{equation}
D_{lm}^{q}\xi^{A,s}_n[\varphi]=\int F[\varphi]%
(\theta)D_{lm}^{q}\xi^{A,s}_n(\theta)d\theta\quad \label{Dxif=}
\end{equation}
with $\xi^{A,s}_n$ of (\ref{xins}), and  in view of (\ref{Dxi}), (\ref{dxi<})
and (\ref{phil<})
\begin{equation}
|D_{lm}^{q}\xi^{A,s}_n[\varphi]|\leq\int|F[\varphi]%
(\theta)||D_{lm}^{q}\xi^{A,s}_n(\theta)|d\theta\leq C_{A}C_\varphi2^{q+1},
\label{Dphi<}
\end{equation}
so that%
\begin{equation}
|D_{lm}^{q}e_{n}(x,s)|\leq\big(c\sqrt{n}(1+|x|)\big)^q.
\label{Dem<}
\end{equation}
Here and in what follows $c$  depends only on $A$ and $\varphi$. This, (\ref{UAUOn}), and  (\ref{DjkF=}) yield
\begin{equation}
|D_{lm}^p\Psi_{lm}(x,t,s)|\leq (c\sqrt{n}(1+|x|+t))^{p+1},\quad
x\in \mathbb{R},\; t>0.  \label{Phijk<}
\end{equation}
Thus, $\Psi_{lm}$ for every $x\in \mathbb{R}$, $t>0$ satisfies  (\ref{al}). Besides, for every $x\in \mathbb{R}$,
 $t>0$ (\ref{kap<}) follows from (\ref{ser<}). Now applying differentiation formula (\ref{difinf}) with $\zeta =W^{(n)}_{lm}$ and $\Phi=\Psi_{lm}$ to every term of the first sum and differentiation formula (\ref{diffga}) to every term of the second sum in the r.h.s. of (\ref{lnZZ}) and taking in
account (\ref{Wjk}), we get (see also (\ref{difinf})):
\begin{align}
\ln Z^A_{n}(x)/\widehat{Z}^A_{n}(x)
=-\frac{x}{2}\int_{0}^{1}\frac{ds}{Z^A_{n}(x,s)}\int
\sum_{p=2}^\infty s^{(p-1)/2}\frac{\kappa_{p+1}%
}{p!}S_p^{(n)}(x,t,s)F[\varphi](t)dt,     \label{lnZZS}
\end{align}%
where
\begin{eqnarray}
\quad
S_p^{(n)}(x,t,s)=\frac{1}{n^{(p+1)/2}}\sum_{l,m=1}^n
\beta_{lm}^{-(p+1)/2} \mathbf{E}\{D_{lm}^p\Psi_{lm}(x,t,s)\}.
\label{Sp}
\end{eqnarray}%
It was shown in \cite{Ly-Pa:11} that in the case of matrix elements (\ref{fjj})  the series in (\ref{lnZZS}) converges uniformly in  $n\in%
\mathbb{N}$, $(t,x)\in K$ for any compact set $K\subset\{(x,t)\in\mathbb{R}^2:t>0\}$. In general case the proof is  almost the same with the obvious
modifications. It is based on (\ref{ser<}), the estimate
\begin{equation}
A_p\leq 2^{p/2}, \quad \forall p \in \mathbb{N},\quad
\label{Ap<}
\end{equation}
following from (\ref{Ap}) and  (\ref{1}), and on uniform bound
\begin{equation}
|S_p^{(n)}(x,t,s)|\leq (C_K)^l, \quad \forall (t,x)\in K,\;n\in \mathbb{N},\quad
s\in[0,1],
\label{ckl}
\end{equation}
which can be obtained with the help of (\ref{DAU}) -- (\ref{dxi<}). Here $C_K$ is an absolute constant depending only on $K$. In view of the uniform convergence of the series, to
make the limiting transition as $n\rightarrow\infty$ in (\ref{lnZZS}) it suffices to find the
limits
\begin{align*}
S_p=\lim_{n\rightarrow\infty}S_p^{(n)}
\end{align*}%
for every fixed $p\in\mathbb{N}$.
We have
\begin{align}
S_p^{(n)}=&\frac{1}{n^{(p+1)/2}}\sum_{l,m=1}^n
\beta_{lm}^{-(p+1)/2}
 \mathbf{E}\Big\{(U*AU)_{lm}D_{lm}^pe^\circ_{n}+pD_{lm}(U*AU)_{lm}D_{lm}^{p-1}e^\circ_{n}
\label{Sp=} \\
 &+\frac{p(p-1)}{2}D^{2}_{lm}(U*AU)_{lm}D_{lm}^{p-2}e^\circ_{n}+(1-\delta_{p2})
 \frac{p(p-1)(p-2)}{6}D^{3}_{lm}(U*AU)_{lm}D_{lm}^{p-3}e^\circ_{n}\notag
 \\
 &+(1-\delta_{p2})(1-\delta_{p3})\sum_{q=0}^{p-4} \Big(%
\begin{array}{ll}
p&  \\
q&
\end{array}%
\Big) D_{lm}^{p-q}(U*A^{(n)}U)_{ml}D_{lm}^{q}e^{\circ}_{n}\Big\}\notag
\\
&\hspace{-0.5cm}=S_{p1}^{(n)}+S_{p2}^{(n)}+\frac{p(p-1)}{2}S_{p3}^{(n)}+(1-\delta_{p2})
 \frac{p(p-1)(p-2)}{6}S_{p4}^{(n)}
+(1-\delta_{p2})(1-\delta_{p3})S_{p5}^{(n)}\notag .
\end{align}%
It follows from (\ref{UAUOn}) --  (\ref{UAUmm}) and (\ref{Dem<})
that
\begin{equation}
S_{p5}^{(n)}=O(n^{-1/2}),\quad n\rightarrow\infty.\label{S5}
\end{equation}
Since
\begin{align}
D_{lm}^{q}e_{n}(x,s)&=e_{n}(x,s)\big(ixD_{lm}\xi^{A,s}_n[\varphi]\big)^{q}
+O(n^{(q-1)/2})\label{Demq}
\\
&=e_{n}(x,s)\big(-x\beta_{lm}\int \widehat{\varphi}
(\theta)(U*C^{(n)}U)_{lm}(\theta)d\theta\big)^{q}
+O(n^{(q-1)/2}),\quad n\rightarrow\infty,\notag
\end{align}
then
 \begin{align*}
 S_{p4}^{(n)}&=\frac{1}{n^{(p+1)/2}}\sum_{l,m=1}^n
\beta_{lm}^{-(p+1)/2}
 \mathbf{E}\big\{D^{3}_{lm}(U*A^{(n)}U)_{lm}
 \\
 &\hspace{2cm}\times\big(-x\beta_{lm}\int \widehat{\varphi}
(\theta)(U*C^{(n)}U)_{lm}(\theta)d\theta\big)^{p-3}e_{n}\big\}+O(n^{-1/2}),\quad
p>3,
\end{align*}
and by (\ref{UAUOn}) -- (\ref{UAUlm}) $S_{p4}^{(n)}=O(n^{-1/2})$, $p>3$.
If $p=3$, then
\begin{align*}
 S_{34}^{(n)}&=\frac{1}{n^2}\sum_{l,m=1}^n
\beta_{lm}^{-2}
 \mathbf{E}\big\{D^{3}_{lm}(U*A^{(n)}U)_{lm}
 e^\circ_{n}(x,s)\big\},
\end{align*}
(compare with $T_{31}^n$ of (\ref{T3})), and in addition to (\ref{UAUOn}) -- (\ref{UAUlm}) we use (\ref{varvC}) to show that $S_{34}^{(n)}=O(n^{-1/2})$.
Thus,
\begin{equation}
S_{p4}^{(n)}=O(n^{-1/2}),\quad n\rightarrow\infty,\quad
p\geq3.\label{S4}
\end{equation}
Consider now $S_{p3}^{(n)}$ of (\ref{Sp=}):
\begin{align*}
S_{p3}^{(n)}&=\frac{1}{n^{(p+1)/2}}\sum_{l,m=1}^n
\beta_{lm}^{-(p+1)/2}
 \mathbf{E}\{D^{2}_{lm}(U*A^{(n)}U)_{lm}D_{lm}^{p-2}e^\circ_{n}\}\notag
\\
&=\frac{1}{n^{(p+1)/2}}\sum_{l,m=1}^n
\beta_{lm}^{(p-5)/2}
 \mathbf{E}\{D^{2}_{lm}(U*A^{(n)}U)_{lm}\big(-x\int \widehat{\varphi}
(\theta)(U*C^{(n)}U)_{lm}(\theta)d\theta\big)^{p-2}e_{n}\}
\\
&+O(n^{-1/2}),\quad n\rightarrow\infty,\quad p>2,
\end{align*}%
where we used (\ref{Demq}).
There arise sums of three types
\begin{align*}
S_{p3}^{1(n)}&=\frac{1}{n^{(p+1)/2}}\sum_{l,m=1}^n
\beta_{lm}^{(p-1)/2}U_{ll}U_{mm}
 (UA^{(n)}U)_{lm}(UC^{(n)}U)^{p-2}_{lm},
\\
S_{p3}^{2(n)}&=\frac{1}{n^{(p+1)/2}}\sum_{l,m=1}^n
\beta_{lm}^{(p-1)/2}U_{lm}U_{lm}
 (UA^{(n)}U)_{lm}(UC^{(n)}U)^{p-2}_{lm},
 \\
 S_{p3}^{3(n)}&=\frac{1}{n^{(p+1)/2}}\sum_{l,m=1}^n
\beta_{lm}^{(p-1)/2}U_{ll}U_{lm}
 (UA^{(n)}U)_{mm}(UC^{(n)}U)^{p-2}_{lm},
\end{align*}
where we omit arguments of $U$ and put $(UC^{(n)}U)_{lm}^{q}=\prod_{j=1}^q(U(t_{j_1})C^{(n)}U(t_{j_2}))_{lm}$.
If $p=2$, then treating $S_{23}^{(n)}$ similar to $T_{21}^{(n)}$ of (\ref{lT2=})
(see (\ref{lT2=}) -- (\ref{T21})) we get $S_{23}^{(n)}=O(n^{-1/2})$,
$n\rightarrow\infty.$
In case $p>2$ we use following from (\ref{AUOn}) -- (\ref{UAUmm}) asymptotic
relations
\begin{align}
&\sum_{l,m=1}^n
|U_{lm}
 ||(UA^{(n)}U)_{lm}|=O(n^{}),\label{UUAUlm}
\\
&\sum_{l,m=1}^n
|(UA^{(n)}U)_{lm}
 ||(UA^{(n)}U)_{lm}|=O(n^{}),\label{lmlm}
\\
&\sum_{l,m=1}^n
|(UA^{(n)}U)_{mm}
 ||(UA^{(n)}U)_{lm}|=O(n\sqrt{n}),\label{lmmm}
\\
&\sum_{m=1}^n
|(UA^{(n)}U)_{mm}
 ||(UA^{(n)}U)_{mm}|=O(n^{}),\label{mmmm}
\end{align}
as $n\rightarrow\infty$. They together with (\ref{UAUOn}) allows to show
that $S_{p3}^{j(n)}$, $j=1,2,3$ are of the order $O(n^{-1/2})$,
$n\rightarrow\infty$,
so that
\begin{equation}
S_{p3}^{(n)}=O(n^{-1/2}),\quad n\rightarrow\infty.\label{Sp3}
\end{equation}
Consider  $S_{p2}^{(n)}$ of (\ref{Sp=}):
\begin{align}
S_{p2}^{(n)}&=\frac{p}{n^{(p+1)/2}}\sum_{l,m=1}^n
\beta_{lm}^{-(p+1)/2}
 \mathbf{E}\{D_{lm}(U*A^{(n)}U)_{lm}(t,s)D_{lm}^{p-1}e^\circ_{n}(x,s)\}\label{Sp2}
 \\
&=-\frac{2px}{n^{(p+1)/2}}\sum_{l,m=1}^n
\beta_{lm}^{-(p-1)/2}
 \mathbf{E}\{\big(U_{ll}*(U*A^{(n)}U)_{mm}+U_{lm}*(U*A^{(n)}U)_{lm}\big)(t,s)\notag
 \\
 &\hspace{8cm}\times D_{lm}^{p-2}
 \big(e_{n}(x,s)D_{lm}\xi^{A,s}_n[\varphi]\big)\}\notag
 \\
 &=S_{p2}^{1(n)}+S_{p2}^{2(n)},
 \notag
\end{align}%
where we used (\ref{DUAUlm}) and  (\ref{De=}).
Since
\begin{align}
D_{lm}^{q}
 \big(e_{n}(x,s)&D_{lm}\xi^{A,s}_n[\varphi]\big)=D_{lm}^{q}
 e_{n}(x,s)\cdot D_{lm}\xi^{A,s}_n[\varphi]\label{De2}
 \\
 &+qD_{lm}^{q-1}
 e_{n}(x,s)\cdot D^{2}_{lm}\xi^{A,s}_n[\varphi]+O(n^{(q-1)/2}),
 \quad n\rightarrow\infty,\notag
   \end{align}
where
\begin{align}
 &D_{lm}\xi^{A,s}_n[\varphi]=i\beta_{lm}\int \label{Dxif}
 (U*^{(n)}CU)_{lm}(\theta,s)F[\varphi](\theta) d\theta,
 \\
 &D^2_{lm}\xi^{A,s}_n[\varphi]=-\beta^2_{lm}\int \label{Dxif}
 \big(U_{ll}\ast
(U\ast
C^{(n)}U)_{mm}+U_{mm}\ast
(U\ast
C^{(n)}U)_{ll}\notag\\
&\hspace{6cm}+2U_{lm }\ast
(U\ast
C^{(n)}U)_{lm}\big)(\theta,s)F[\varphi](\theta) d\theta,\notag
   \end{align}
then putting (\ref{De2}) with $q=p-2$ in $S_{p2}^{2(n)}$ of (\ref{Sp2}) and applying (\ref{Dem<}), (\ref{lmlm}), and (\ref{mmmm}) we get
\begin{equation}
S_{p2}^{2(n)}=O(n^{-1/2}),\quad n\rightarrow\infty,\label{Sp3}
\end{equation}
and \begin{align*}
S_{p2}^{1(n)}=&-2px\int F[\varphi](\theta) d\theta\frac{1}{n^{(p+1)/2}}\sum_{l,m=1}^n
\beta_{lm}^{-(p+1)/2}
 \mathbf{E}\{\big(U_{ll}*(U*A^{(n)}U)_{mm}\big)(t,s)
 \\
 &\times \big[D_{lm}^{p-2}
 e_{n}(x,s)\cdot i(U*C^{(n)}U)_{lm}(\theta,s)\notag
 \\
 &\quad-(p-2)D_{lm}^{p-3}
 e_{n}(x,s)\cdot \big(U_{ll}\ast
(U\ast
C^{(n)}U)_{mm}+U_{mm}\ast
(U\ast
C^{(n)}U)_{ll}\big)(\theta,s)\big]\}
\\
&\quad+O(n^{-1/2}),\quad n\rightarrow\infty.\notag
 \end{align*}%
It follows from (\ref{Demq}) and (\ref{UUAUlm}) -- (\ref{mmmm}) that $S_{p2}^{1(n)}$ does
not vanish only  if $p=2$ or $p=3,$ so that putting $e_{n}(x,s)=Z^{A}_{n}(x,s)+e^\circ_{n}(x,s)$
and using ( \ref{varvC}) and (\ref{varw}), we get
\begin{align}
S_{p2}^{1(n)}=&xZ^{A}_{n}(x,s)\int
\Big[-\frac{4i\delta_{p2}}{n^{3/2}}\sum_{l,m=1}^n
 \mathbf{E}\{\big(U_{ll}*(U*A^{(n)}U)_{mm}\big)(t,s)
(U*C^{(n)}U)_{lm}(\theta,s)\}\notag
 \\
 &\quad+\frac{6\delta_{p3}}{n^{2}}\sum_{l,m=1}^n
 \mathbf{E}\{\big(U_{ll}*(U*A^{(n)}U)_{mm}\big)(t,s) \Big(U_{ll}\ast
(U\ast
C^{(n)}U)_{mm}\label{Sp21}
\\
&\hspace{5cm}+U_{mm}\ast
(U\ast
C^{(n)}U)_{ll}\Big)(\theta,s)\}\Big]F[\varphi](\theta) d\theta+O(n^{-1/2}),\notag
 \end{align}%
as $n\rightarrow\infty$.
Such expressions were considered while proving Theorem \ref{t:Cov} (see Lemma
\ref{l:T2T3}). Treating $S_{p2}^{1(n)}$, $p=2,3$ in the same way and using (\ref{limvC}), (\ref{limw})  and (\ref{vv}) -- (\ref{vvv}), we get
\begin{align}
\lim_{n\rightarrow\infty}-\frac{x}{2}\int_{0}^{1}\frac{ds}{Z^A_{n}(x,s)}\int
\Big[\frac{\kappa_{3}\sqrt{s}}{2}S_{22}^{1(n)}(x,t,s)+&\frac{\kappa_{4}s}{6}
S_{32}^{1(n)}(x,t,s)\Big]F[\varphi](t)dt\label{k3k4}
\\
&=-\Big(\frac{2}{3}C_{\kappa _{3}}[\varphi ,\varphi] +C_{\kappa _{4}}[\varphi,\varphi]\Big)x^{2}/2     \notag
\end{align}%
 with $C_{\kappa _{3}}[\varphi ,\varphi]$,
$C_{\kappa _{4}}[\varphi,\varphi]$ of (\ref{C3}) -- (\ref{C4}) (see also
(\ref{lnGW}), ( \ref{lnZZS})).

At last consider $S_{p1}^{(n)}$ of (\ref{Sp=}):
\begin{align}
S_{p1}^{(n)}&=\frac{1}{n^{(p+1)/2}}\sum_{l,m=1}^n
\beta_{lm}^{-(p+1)/2}
 \mathbf{E}\{(U*A^{(n)}U)_{lm}(t,s)D_{lm}^pe_{n}^\circ(x,s)\}\label{Sp1=}
 \\
 &=\frac{ix}{n^{(p+1)/2}}\sum_{l,m=1}^n
\beta_{lm}^{-(p+1)/2}
 \mathbf{E}\{(U*A^{(n)}U)_{lm}(t,s)D_{lm}^{p-1}\big(e_{n}(x,s)D_{lm}\xi^{A,s}_n[\varphi]\big)\}\notag
  \\
 &=\frac{ix}{n^{(p+1)/2}}\sum_{l,m=1}^n
\beta_{lm}^{-(p+1)/2}
 \mathbf{E}\{(U*A^{(n)}U)_{lm}(t,s)\big[D_{lm}^{p-1}e_{n}(x,s)\cdot D_{lm}\xi^{A,s}_n[\varphi]\notag
   \\
 &\hspace{2cm}+(p-1)D_{lm}^{p-2}e_{n}(x,s)\cdot D^2_{lm}\xi^{A,s}_n[\varphi]
 \big]\}+O(n^{-1/2})\notag
 \\
 &=S_{p1}^{1(n)}+S_{p1}^{2(n)}+O(n^{-1/2}),\quad n\rightarrow\infty,\notag
  \end{align}%
where we used (\ref{De=}), (\ref{De2}), and then (\ref{dxi<}), (\ref{Dem<}),
and (\ref{UUAUlm}) -- (\ref{mmmm}) to estimate the vanishing term. It follows from (\ref{Demq}) and (\ref{lmlm}) -- (\ref{mmmm}) that if $p>2,$
then
 \begin{equation}
S_{p1}^{2(n)}=O(n^{-1/2}),\quad n\rightarrow\infty,\quad p>2.\label{Sp3}
\end{equation}
If $p=2,$  then
\begin{align*}
S^{2(n)}_{21} &=-ix\int\frac{1}{n^{3/2}}\sum_{l,m=1}^n
 \mathbf{E}\big\{\notag
  e_{n}(x,s)(U*A^{(n)}U)_{lm}(t,s) \big(U_{ll}\ast
(U\ast
C^{(n)}U)_{mm}\notag
\\
&\hspace{6cm}+U_{mm}\ast
(U\ast
C^{(n)}U)_{ll}\big)(\theta,s)\big\}F[\varphi](\theta) d\theta+O(n^{-1/2}),
\end{align*}%
and  similar to (\ref{k3k4})
\begin{align}
\lim_{n\rightarrow\infty}-\frac{x}{2}\int_{0}^{1}\frac{ds}{Z^A_{n}(x,s)}\int
\frac{\kappa_{3}\sqrt{s}}{2}S_{21}^{2(n)}(x,t,s)F[\varphi](t)dt
=-\Big(\frac{1}{3}C_{\kappa _{3}}[\varphi  ,\varphi] \Big)x^{2}/2.     \label{k3}
\end{align}%
Using (\ref{Demq}) with $q=p-1$ and (\ref{Dxif}) we get for $S_{p1}^{1(n)}$ of (\ref{Sp1=}):
\begin{align}
S^{1(n)}_{p1}
 =&\frac{1}{n^{(p+1)/2}}\sum_{l,m=1}^n
\beta_{lm}^{(p-1)/2}
 \mathbf{E}\{(U*A^{(n)}U)_{lm}(t,s)e_{n}(x,s) \label{Sp11=}
   \\
 &\quad\times\big(-x\int \widehat{\varphi}
(\theta)(U*C^{(n)}U)_{mm}(\theta)d\theta\big)^{p}\}+O(n^{-1/2}),\quad n\rightarrow\infty,\notag
  \end{align}%
where we estimate the
vanishing term with the help of  (\ref{lmlm}) and  (\ref{mmmm}). Putting here
\begin{equation*}
\beta^{(p-1)/2}_{lm}=1+\delta_{lm}
(2^{(1-p)/2}-1)
\end{equation*}
and $e_{n}(x,s)=Z^A_{n}(x,s)+e^\circ_{n}(x,s) $, and then applying first
parts of (\ref{limg1}) -- (\ref{limg2}), we get
\begin{align*}
S^{1(n)}_{p1}
 =&\frac{Z^{A}_{n}(x,s)}{n^{(p+1)/2}}\sum_{l,m=1}^n
 \mathbf{E}\{(U*A^{(n)}U)_{lm}(t,s)\big(-x\int \widehat{\varphi}
(\theta)(U*C^{(n)}U)_{lm}(\theta)d\theta\big)^{p}\} \label{Sp11=}
   \\
  &+2(2^{(1-p)/2}-1)\frac{Z^{A}_{n}(x,s)}{n^{(p+1)/2}}\sum_{m=1}^n
 \mathbf{E}\{(U*A^{(n)}U)_{mm}(t,s)
 \\
 &\hspace{3cm}\times\big(-x\int \widehat{\varphi}
(\theta)(U*C^{(n)}U)_{mm}(\theta)d\theta\big)^{p}\}
+O(n^{-1/2}),\quad n\rightarrow\infty.\notag
  \end{align*}%
This and second parts of (\ref{limg1}) -- (\ref{limg2})
  yield
for $p\geq 2$
\begin{equation}
\lim_{n\rightarrow\infty}-\frac{x}{2}\frac{\kappa_{p+1}%
}{p!}\int_{0}^{1}\frac{ds}{Z^A_{n}(x,s)}\int
s^{(p-1)/2}S^{1(n)}_{p1}
 (x,t,s)F[\varphi](t)dt
=\frac{\kappa_{p+1} A_{p+1} }{(p+1)!}(ix^*)^{p+1}
\label{Sp+1}
\end{equation}
with $A_p$ and $x^*$ defined in (\ref{Ap}) and (\ref{x}).
Now putting (\ref{k3k4}), (\ref{k3}), and (\ref{Sp+1}) in (\ref{lnZZ}) we get
(\ref{Sp+1}) and finish the proof of the theorem under condition (\ref{Sp+1}).

The case of $\varphi\in E=\{\psi:\int(1+|t|)^{4}|\widehat{\psi }%
(t)|dt<\infty\}$ can be obtained via a standard approximation procedure.
Indeed, since the set $D=\{\varphi:\int|\widehat{\varphi\ }%
(t)||t|^{l}dt<C_\varphi l!,\;\forall l\in\mathbb{N}\}$ is big enough (in
particular, it contains functions $e^{-x^2}P_m(x)$, where $P_m(x)$ is a
polynomial), then for any  $\varphi \in E$ there exists a sequence $%
\{\varphi _{k}\}\subset D$, such  that
\begin{equation}
\lim_{k\rightarrow \infty }\int_{-2w}^{2w}|\varphi (\lambda )-\varphi
_{k}(\lambda )|d\lambda=0.  \label{fktof}
\end{equation}%
Denote for the moment the characteristic functions of $\xi^{A}_n[\varphi]$ and $\xi^{A}[\varphi]$ as $Z_{n}[\varphi ]$ and $Z[\varphi ]$, to make explicit their
dependence on  $\varphi $. We have then for any $\varphi \in E$
\begin{align}
|Z_{n}[\varphi ]-Z[\varphi ]|&\leq |Z_{n}[\varphi ]-Z_{n}[\varphi
_{k}]|+|Z_{n}[\varphi _{k}]-Z[\varphi _{k}]|+|Z[\varphi _{k}]-Z[\varphi ]|
\notag \\
&:=T_{nk}^{(1)}+T_{nk}^{(2)}+T_{nk}^{(3)}.  \label{chaZ}
\end{align}%
The second term of the r.h.s. vanishes after the limit $n\rightarrow \infty $
in view of the above proof, since $\varphi _{k}\in D$. For the first term we
have from (\ref{Znxj}) and the Schwarz inequality that
\begin{eqnarray*}
|T_{nk}^{(1)}| &=|\mathbf{E}\{e^{ix\xi^{A\circ}_n[\varphi]}-e^{ix\xi_n^{A\circ}[\varphi_k]}\}|\leq|x|
\big(n\mathbf{Var}\{\xi_n^{A}[\varphi_k]\}\big)%
^{1/2},\quad \psi_k=\varphi-\varphi_k,
\end{eqnarray*}%
and then Theorem \ref{t:Cov} implies that
\begin{eqnarray*}
\limsup_{n\rightarrow \infty }|T_{nk}^{(1)}| &\leq &|x|(V_W[\psi_k])^{1/2}.
\end{eqnarray*}%
Since $V_W$ of (\ref{VW}) is continuous with respect to the $L^{1}$
convergence, then in view of (\ref{fktof}) $T_{nk}^{(1)}$ vanishes after the
subsequent limits $n\rightarrow \infty$, $k\rightarrow \infty $.

At last, we have by (\ref{ser<}), (\ref{Ap<}), and the continuity of the r.h.s.  of (\ref%
{ln}) with respect to the $L^{1}$ convergence, that the third term of (\ref%
{chaZ}) vanishes after the limit $k\rightarrow \infty .$ Thus, we
proved  Theorem under condition (\ref{F4<}).
\end{proof}

\section{Auxiliary results}\label{s:aux}

\begin{lemma}
\label{l:main} Consider  matrix $A^{(n)}$,
satisfying
 (\ref{1}) -- (\ref{TA}), $C^{(n)}=A^{(n)}+A^{(n)T}$, and  a unitary matrix
\begin{equation*}
U (t)=U^{(n)}(t)=e^{itM^{(n)}},
\end{equation*}
 corresponding to the Wigner matrix  $M^{(n)}$     of (\ref{MW}) -- (\ref{Wmom12}).
 Denote
 \begin{align*}
& U^j=U(t_j),
\\
&\overline{t^{(p)}}=(t_1,...,t_p),
 \end{align*}
and define
\begin{align}
&\xi^{A}_n(t)=\Tr A^{(n)}U(t),  \notag \\
&\eta^{A}_n(t_1,t_2)=n^{-3/2}\sum_{l,m=1}^n(U^{1}A^{(n)}U^{2})_{lm},  \label{eta} \\
&v^{I}_{n}(t_1,t_2,t_{3})=n^{-1}\sum_{m=1}^nU^{1}_{mm}(U^{2}A^{(n)}U^{3})_{mm},\quad
\label{vI} \\
&v^{C}_{n}(\overline{t^{(4)}})=n^{-1}\sum_{m=1}^n(U^{1}C^{(n)}U^{2})_{mm}(U^{3}A^{(n)}U^{4})_{mm},\quad
\label{vC} \\
&\omega_{n}(\overline{t^{(5)}})=n^{-3/2}\sum_{l,m=1}^nU^{1}_{ll}(U^{2}A^{(n)}U^{3})_{mm}(U^{4}C^{(n)}U^{5})_{lm},\quad
\label{omega} \\
&\gamma^{(1)}_{n}(\overline{t^{(2p+2)}})=n^{-(p+1)/2}\sum_{l,m=1}(U^{1}A^{(n)}U^{2})_{lm}\prod_{j=2}^{p+1} (U^{2j-1}C^{(n)}U^{{2j}})_{lm},\quad p\geq2,
\label{gam1}
\\
&\gamma^{(2)}_{n}(\overline{t^{(2p+2)}})=n^{-(p+1)/2}\sum_{m=1}(U^{1}A^{(n)}U^{2})_{mm}\prod_{j=2}^{p+1} (U^{2j-1}C^{(n)}U^{{2j}})_{mm},\quad p\geq2,
\label{gam2}
\end{align}
  and put $$\overline{f}=\mathbf{E}\{f\}.$$
Then we have
under  conditions of Theorem \ref{t:Cov}:
\begin{align}
&\text{(i)} \;\mathbf{Var}\{\xi^{A}_n(t)\} \leq c(1+|t|)^8,\hspace{2.1cm}
\lim_{n\rightarrow\infty}\overline{\xi}^{A}_n(t)=T_{A}\cdot v(t),  \label{xiv} \\
&\text{(ii)} \;\mathbf{Var}\{\eta^{A}_n(t_1,t_2)\}=O(n^{-1})
\hspace{1.9cm}
\lim_{n\rightarrow\infty}\overline{\eta}^{A}_n(t_1,t_2)=K^{'(2)}_{A}\cdot v(t_{1})v(t_{2}),\; \label{Vareta} \\
&\text{(iii)} \;\mathbf{Var}\{v^{I}_{n}(t_1,t_2,t_{3})\}=O(n^{-1}),\;
\hspace{1.1cm}
\lim_{n\rightarrow\infty}\overline{v}^{I}_{n}(t_1,t_2,t_{3})=T_{A}\cdot v(t_{1})v(t_{2}+t_{3}), \label{limvI} \\
&\text{(iv)} \;\mathbf{Var}\{v^{C}_{n}(\overline{t^{(4)}})\}=O(n^{-1}),\;\label{varvC}
\\
&\hspace{4cm}\lim_{n\rightarrow\infty}\overline{v}^{C}_{n}(\overline{t^{(4)}})=2K_{A}^{(3)}\prod_{j=1}^{4} v(t_j)+2T^{2}_{A}\cdot v(t_{1}+t_{2})v(t_{3}+t_{4}), \label{limvC} \\
&\text{(v)}\; \mathbf{Var}\{\omega_{n}(\overline{t^{(5)}})\}=O(n^{-1/2}),\;  \label{varw}
\\
&\hspace{1cm}\lim_{n\rightarrow\infty}\overline{\omega}_n(\overline{t^{(5)}})=
(K_{A}^{(1)}-K_{A}^{(2)})\prod_{j=1}^5v(t_j)+K_{A}^{(2)}v(t_1)v(t_4)v(t_5)v(t_2+t_3),
\label{limw}
\\
& \text{(vi )}
\mathbf{Var}\{\gamma^{(1)}_{n}(\overline{t^{(2p+2)}})\}=O(n^{-1/2}),\;
\hspace{1.3cm}\lim_{n\rightarrow\infty}\overline{\gamma}^{(1)}_{n}(\overline{t^{(2p+2)}})=K_{A}^{(4)}\prod_{j=1}^{p+1}
v(t_j), \label{limg1}
\\
&\text{(vii)}\mathbf{Var}\{\gamma^{(2)}_{n}(\overline{t^{(2p+2)}})\}=O(n^{-1/2}),\;
\hspace{1.3cm}\lim_{n\rightarrow\infty}\overline{\gamma}^{(2)}_{n}(\overline{t^{(2p+2)}})=K_{A}^{(5)}\prod_{j=1}^{p+1}
v(t_j),  \label{limg2}
\end{align}
where $O(n^\alpha)$, $n\rightarrow\infty$, can depend on $\overline{t^{(p)}}$, $v$ is defined in (\ref{vt}), $K_{A}^{(j)}$, $j=1,2,3$
are defined in (\ref{K1}) -- (\ref{K3}), and
\begin{align}
& K^{'(2)}_{A}=\lim_{n\rightarrow \infty
}n^{-3/2}\sum_{l,m=1}^{n}A^{(n)}_{lm},   \label{K22} \\
& K^{(4)}_{A}=\lim_{n\rightarrow \infty
}n^{-(p+1)/2}\sum_{l,m=1}^{n}A^{(n)}_{lm}(C^{(n)}_{lm})^{p},   \label{K4} \\
&K^{(5)}_{A}=\lim_{n\rightarrow \infty
}n^{-(p+1)/2}\sum_{m=1}^{n}A^{(n)}_{mm}(C^{(n)}_{mm})^{p}.\label{K5}
\end{align}
\end{lemma}

\begin{remark}
All statements of the lemma remain valid under
conditions of Theorem \ref{t:clt}.
\end{remark}
\begin{proof}
{\bf GOE case.} Firstly we prove the lemma supposing that matrix $M^{(n)}$ belongs to the GOE.
Statement
{\bf (i)} in GOE case was proved in Lemma \ref{GOE}.

\medskip

{\bf (ii)}  We have by Poincar\'{e} inequality (\ref{Nash})
  \begin{align*}
\mathbf{Var}\{\eta^{A}_n(t_{1},t_{2})\} &\leq\frac{w^{2}}{n^{4}} \sum_{1\leq j\leq k \leq
n} \beta^{-1}_{jk}\mathbf{E}\big\{\big|D_{jk}\sum_{l,m=1}^n(U^{1}A^{(n)}U^{2})_{lm}\big|^2\big\}.
\end{align*}
This and  (\ref{DUAU}) show that it suffices to estimate
\begin{align*}
T_{n}=\frac{1}{n^{4}}
 \sum_{j,k=1}^n \big|\sum_{l,m=1}^nU^1_{lj}(U^{2}A^{(n)}U^{3})_{km}\big|^2.
\end{align*}
We have
\begin{align*}
T_{n}&=\frac{1}{n^{4}}
 \sum_{j,k=1}^n \sum_{l,l',m,m'=1}^nU^1_{lj}\overline
 {U}^1_{jl'}(\overline{U}^{3}A^{(n)T}\overline{U^{2}})_{mk}(U^{2}A^{(n)}U^{3})_{km}
 \\
&=\frac{1}{n^{3}}
 \sum_{m,m'=1}^n(\overline{U}^{3}A^{(n)T}A^{(n)}U^{3})_{mm'} =\frac{1}{n^{3}}
 \Big|\sum_{m,p=1}^n(A^{(n)}U^{3})_{pm}\Big|^{2}\leq\frac{1}{n^{2}}\Tr
AA^{(n)T}=O(n^{-1}),
\end{align*}
hence,
\begin{align*}
\mathbf{Var}\{\eta^{A}_n(t_{1},t_{2})\}
=O(n^{-1}),\quad n\rightarrow\infty.
\end{align*}
Now
 applying  Duhamel formula (\ref{Duh}) and
differentiation formulas (\ref{diffga}), (\ref{DUAUlm}), and then  estimating
the error terms with the help of (\ref{vGxi<}), one can get
 \begin{align}
\overline{\eta}^{A}_n(t_{1},t_{2})=&
n^{-3/2}\sum_{l,m=1}^n\mathbf{E}\{(A^{(n)}U(t_{2}))_{lm}\}\label{etaeq}
\\
&-{w^{2}}
\int_{0}^{t_{1}}dt_{3}\int_{0}^{t_{3}}\overline{v}_{n}(t_{3}-t_{4})
\overline{\eta}^{A}_n(t_{3},t_{2})dt_{4}+o(1),\quad n\rightarrow\infty, \notag
\end{align}
where by (\ref{EUjk}) and (\ref{vt})
\begin{align*}
\lim_{n\rightarrow\infty}n^{-3/2}\sum_{l,m=1}^n\mathbf{E}\{(A^{(n)}U(t_{2}))_{lm}\}
=\lim_{n\rightarrow\infty}\overline{v}_n(t_{2})n^{-3/2}\sum_{l,m=1}^nA^{(n)}_{lm}=K^{'(2)}_{A}v(t_{2})
\end{align*}
 with $K^{'(2)}_{A}$ of (\ref{K22}). Thus, we have for $\eta^A=\lim_{n\rightarrow\infty}\overline{\eta}^{A}_n$:
 \begin{align*}
\eta^A(t_{1},t_{2})
+{w^{2}}
\int_{0}^{t_{1}}dt_{3}\int_{0}^{t_{3}}v(t_{3}-t_{4})\eta^A(t_{4},t_{2})dt_{4}=K^{'(2)}_{A}v(t_{1}), \end{align*}
and by (\ref{F1})
 \begin{align}
\eta^A(t_{1},t_{2})=
K^{'(2)}_{A}v(t_{1})v(t_{2}), \label{etaK}
\end{align}
so (\ref{Vareta}) is proved.

\medskip

{\bf (iii)} Putting $U_{mm}=U^\circ_{mm}+\overline{U}_{mm}$ and using (\ref{EUjk})
we get
\begin{align}
\overline{v}^{I}_{n}(t_1,t_2,t_{3})
=\overline{v}_n(t_{1})\overline{\xi}_n^A(t_{2}+t_{3})+\overline{r}_n,
\quad r_n=n^{-1}\sum_{m=1}^n(U^{1}_{mm})^\circ(U^{2}A^{(n)}U^{3})_{mm}. \label{vI=}  \end{align}
By the Schwarz inequality and (\ref{UAUmm})
\begin{align}
| r_n|&\leq\Big(\sum_{m=1}^n|(U^{2}A^{(n)}U^{3})_{mm}|^2\Big)^{1/2}
\Big(\sum_{m=1}^n|(U^{1}_{mm})^\circ|^2\Big)^{1/2}\notag
\\
&= O(n^{-1/2})\Big(\sum_{m=1}^n|(U^{1}_{mm})^\circ|^2\Big)^{1/2}. \label{rnv} \end{align}
It follows from (\ref{vGxi<}) that
\begin{align}
\mathbf{E}\Big\{\sum_{m=1}^n|(U^{1}_{mm})^\circ|^2\Big\}
=O(1),\quad n\rightarrow\infty. \label{mm0}
\end{align}
  This, (\ref{rnv}), and the Schwarz inequality for expectations yield
\begin{align}
| \overline{r}_n|
=O(n^{-1/2}),\quad n\rightarrow\infty. \label{rnO}
\end{align}
Now (\ref{xiv}),  (\ref{vI=}), and (\ref{rnO}) give
\begin{align*}
\lim_{n\rightarrow\infty}\overline{v}^{I}_{n}(t_1,t_2,t_{3})=T_{A}\cdot v(t_{1})v(t_{2}+t_{3}).
\end{align*}
Besides, we have for $V_n=\mathbf{Var}\{v^{I}_{n}\}$:
\begin{align*}
V_n=\mathbf{E}\big\{n^{-1}\sum_{m=1}^nU^{1}_{mm}(U^{2}A^{(n)}U^{3})_{mm}\cdot
v^{I\circ}_{n}\big\}=\overline{v}_n(t_{1})\mathbf{E}\big\{n^{-1}{\xi}_n^A(t_{2}+t_{3})v^{I\circ}_{n}\big\}
+\mathbf{E}\big\{r_{n}v^{I\circ}_{n}\big\}
\end{align*}
with $r_n$ of (\ref{vI=}).
It follows from the Schwarz inequality, (\ref{vGxi<})
and (\ref{rnv}) -- (\ref{rnO}), that
\begin{align*}
V_n\leq O(n^{-1/2})V_n^{1/2},\quad n\rightarrow\infty.
\end{align*}
This proves (\ref{limvI}).

\medskip

{\bf (iv)}  The proof of (\ref{varvC}) repeats with the obvious modifications
 that one of (\ref{Vareta}).
 Let us prove (\ref{limvC}). Applying  Duhamel formula (\ref{Duh}),
differentiation formulas (\ref{diffga}), (\ref{DAU}) -- (\ref{DUAUlm}), and then  estimating
the error terms with the help of (\ref{vGxi<}), one can get
(cf (\ref{Fn}) and (\ref{etaeq})) \begin{align}
n^{-1}\sum_{m=1}^n\mathbf{E}\{&(U^{1}A^{(n)}U^{2})_{mm}(U^{3}C^{(n)}U^{4})_{mm}\}
=n^{-1}\sum_{m=1}^n\mathbf{E}\{(A^{(n)}U^{2})_{mm}(U^{3}C^{(n)}U^{4})_{mm}\}\label{vnCe}
\\
&-{w^{2}}
\int_{0}^{t_{1}}dt_{5}\int_{0}^{t_{5}}\mathbf{E}\Big\{v_{n}(t_{5}-t_{6})n^{-1}
\sum_{m=1}^n(U^{6}A^{(n)}U^{2})_{mm}(U^{3}C^{(n)}U^{4})_{mm}\Big\}dt_{6}\notag \\
&-{w^{2}}
\int_{0}^{t_{1}}dt_{5}\int_{0}^{t_{5}}\mathbf{E}\Big\{n^{-1}\xi_{n}^{A}(t_{5}+t_{6})n^{-1}
\sum_{m=1}^nU_{mm}(t_{2}-t_{6})(U^{3}C^{(n)}U^{4})_{mm}\Big\}dt_{6}\notag
\\
&+o(1),\quad n\rightarrow\infty, \notag
\end{align}
and
 \begin{align}
n^{-1}\sum_{m=1}^n\mathbf{E}\{&(A^{(n)}U^{2})_{mm}(U^{3}C^{(n)}U^{4})_{mm}\}
=n^{-1}\sum_{m=1}^nA^{(n)}_{mm}\mathbf{E}\{(U^{3}C^{(n)}U^{4})_{mm}\}\label{Gne}
\\
&-{w^{2}}
\int_{0}^{t_{2}}dt_{5}\int_{0}^{t_{5}}\mathbf{E}\Big\{v_{n}(t_{5}-t_{6})n^{-1}
\sum_{m=1}^n(A^{(n)}U^{6})_{mm}(U^{3}C^{(n)}U^{4})_{mm}\Big\}dt_{6}\notag \\
&+o(1),\quad n\rightarrow\infty, \notag
\end{align}
and
also\begin{align}
n^{-1}\sum_{m=1}^nA^{(n)}_{mm}&\mathbf{E}\{(U^{3}C^{(n)}U^{4})_{mm}\}
=n^{-1}\sum_{m=1}^nA^{(n)}_{mm}\mathbf{E}\{(C^{(n)}U^{4})_{mm}\}\label{HnCe}
\\
&-{w^{2}}
\int_{0}^{t_{3}}dt_{5}\int_{0}^{t_{5}}\mathbf{E}\Big\{v_{n}(t_{5}-t_{6})n^{-1}
\sum_{m=1}^nA^{(n)}_{mm}(U^{6}C^{(n)}U^{4})_{mm}\Big\}dt_{6}\notag \\
&-{w^{2}}
\int_{0}^{t_{3}}dt_{5}\int_{0}^{t_{4}}\mathbf{E}\Big\{n^{-1}\xi_{n}^{C}(t_{5}+t_{6})n^{-1}
\sum_{m=1}^nA^{(n)}_{mm}U_{mm}(t_4-t_6 )\Big\}dt_{6}\notag
\\
&+o(1),\quad n\rightarrow\infty, \notag
\end{align}
where by (\ref{EUjk})
\begin{equation*}
n^{-1}\sum_{m=1}^nA^{(n)}_{mm}\mathbf{E}\{(C^{(n)}U^{4})_{mm}\}
={\overline{v}}_n(t_{4})n^{-1}\sum_{m=1}^nA^{(n)}_{mm}C^{(n)}_{mm},
\end{equation*}
so that
\begin{equation}
\lim_{n\rightarrow\infty}n^{-1}\sum_{m=1}^nA^{(n)}_{mm}\mathbf{E}\{(C^{(n)}U^{4})_{mm}\}=2 D_{A}\cdot v(t_{4}), \quad D_A=\lim_{n\rightarrow\infty}n^{-1}\sum_{m=1}^n(A^{(n)}_{mm})^2. \label{DA}
\end{equation}
Denote
\begin{align*}
{v}^{C}(\overline{t^{(4)}})&=\lim_{n\rightarrow\infty}\overline{v}^{C}_{n}(\overline{t^{(4)}}),
\\
G(t_{2},t_{3},t_{4})&=\lim_{n\rightarrow\infty}n^{-1}\sum_{m=1}^n\mathbf{E}\{(A^{(n)}U^{2})_{mm}(U^{3}C^{(n)}U^{4})_{mm}\},
\\
H(t_{3},t_{4})&=\lim_{n\rightarrow\infty}n^{-1}\sum_{m=1}^nA^{(n)}_{mm}\mathbf{E}\{(U^{3}C^{(n)}U^{4})_{mm}\}.
\end{align*}
(More accuratly, it can be shown that there exist corresponding convergent subsequences. But all these subsequences have the same limits, which are
unique solutions of the system of integral equations below. Hence, we can
write limits of whole sequences.) It follows from (\ref{vnCe}) -- (\ref{DA}), (\ref{xiv}) and (\ref{limvI})
that $v^{C}$, $G$, and $H$ satisfy the system of integral equations:
\begin{align*}
{v}^{C}(\overline{t^{(4)}})+&{w^{2}}
\int_{0}^{t_{1}}dt_{5}\int_{0}^{t_{5}}v(t_{5}-t_{6}){v}^{C}(t_{6},t_{2},t_{3},t_{4})dt_{6}\notag \\
&=G(t_{2},t_{3},t_{4})-2{w^{2}}T_A^2v(t_{3}+t_{4})
\int_{0}^{t_{1}}dt_{5}\int_{0}^{t_{5}}v(t_{5}+t_{6})
v(t_{2}-t_{6})dt_{6},
\\
G(t_{2},t_{3},t_{4})+&{w^{2}}
\int_{0}^{t_{2}}dt_{5}\int_{0}^{t_{5}}v(t_{5}-t_{6})G(t_{6},t_{3},t_{4})dt_{6}=H(t_{3},t_{4}),
\\
H(t_{3},t_{4})+&{w^{2}}
\int_{0}^{t_{3}}dt_{5}\int_{0}^{t_{5}}v(t_{5}-t_{6})H(t_{6},t_{4})dt_{6}\\
&=2 D_{A}\cdot v(t_{4})-2{w^{2}}T_A^2
\int_{0}^{t_{3}}dt_{5}\int_{0}^{t_{4}}v(t_{5}+t_{6})
v(t_{4}-t_{6})dt_{6}.
\end{align*}
 Solving the equations with the help of (\ref{F1}) -- (\ref{F2}) we get
\begin{align*}
H(t_{3},t_{4})&=2K^{(3)}_Av(t_{3})v(t_{4})+{2}T_A^2v(t_{3}+t_{4}),
\\
G(t_{2},t_{3},t_{4})&=v(t_{2})H(t_{3},t_{4}),
\\
{v}^{C}(\overline{t^{(4)}})&=v(t_{1})G(t_{2},t_{3},t_{4})+{2}T_A^2(v(t_{1}+t_{2})-v(t_{1})v(t_{2}))v(t_{3}+t_{4}),
\end{align*}
so that
\begin{align*}
{v}^{C}(\overline{t^{(4)}})
=2K^{(3)}_A\prod_{j=1}^4v(t_j)+{2}T_A^2v(t_{1}+t_{2})v(t_{3}+t_{4}),
\end{align*}
and (\ref{limvC}) is proved.

\medskip

{\bf (v)} Similar to (\ref{vI=}) -- (\ref{rnO}) we have
\begin{align}
\overline{\omega}_n(\overline{t^{(5)}})&=\overline{v}_n(t_{1})\overline{\Gamma}_{n}(t_{2},t_{3},t_{4},t_{5})
+\overline{r}_n, \label{om=}
\end{align}
where
\begin{align}
&\Gamma_{n}(t_{2},t_{3},t_{4},t_{5})
=n^{-3/2}\sum_{l,m=1}^n(U^{2}A^{(n)}U^{3})_{mm}(U^{4}C^{(n)}U^{5})_{lm}, \label{Gam=}
\\
&r_n(\overline{t^{(5)}})=n^{-3/2}\sum_{l,m=1}^n(U^{1}_{ll})^\circ(U^{2}A^{(n)}U^{3})_{mm}
(U^{4}C^{(n)}U^{5})_{lm},\notag
\end{align}
and by (\ref{norU}), the Schwarz inequality,  (\ref{UAUmm}), and (\ref{mm0})
\begin{align}
| \overline{r}_n|&\leq n^{-3/2}||C^{(n)}||\Big(\mathbf{E}\Big\{\sum_{l=1}^n|(U^{1}_{ll})^\circ|^2\Big\}\Big)^{1/2}
\Big(\mathbf{E}\Big\{\sum_{m=1}^n|(U^{2}A^{(n)}U^{3})_{mm}|^2\Big\}\Big)^{1/2}
\notag
\\
&= O(n^{-1/2}),\quad n\rightarrow\infty. \label{rnom=}
\end{align}
Applying  Duhamel formula (\ref{Duh}),
differentiation formulas (\ref{diffga}), (\ref{DAU}) -- (\ref{DUAUlm}), and then  estimating
the error terms with the help of (\ref{vGxi<}), one can get for $\Gamma _{n}$ (cf (\ref{vnCe})
-- (\ref{HnCe})):
 \begin{align*}
\overline{\Gamma}_{n}(t_{2},t_{3},t_{4},t_{5})=&\overline{B}_{n}(t_{3},t_{4},t_{5})
\\
&-{w^{2}}
\int_{0}^{t_{2}}dt_{6}\int_{0}^{t_{6}}\mathbf{E}\big\{v_{n}(t_{6}-t_{7})\Gamma_{n}(t_{7},t_{3},t_{4},t_{5})\big\}dt_{7}\notag \\
&-{w^{2}}
\int_{0}^{t_{2}}dt_{6}\int_{0}^{t_{3}}\mathbf{E}\big\{n^{-1}\xi_{n}^{A}(t_{6}+
t_{7})D_{n}(t_{3}-t_{7},t_{4},t_{5})\big\}dt_{7}\notag
\\
&+o(1),\quad n\rightarrow\infty, \notag
\end{align*}
where
\begin{align*}
&B_{n}(t_{3},t_{4},t_{5})=n^{-3/2}\sum_{l,m=1}^n(A^{(n)}U^{3})_{mm}(U^{2}C^{(n)}U^{3})_{lm}, \\
&D_{n}(\tau,t_{4},t_{5})=n^{-3/2}\sum_{l,m=1}^nU^{}_{mm}(\tau)(U^{2}C^{(n)}U^{3})_{lm}.\notag \end{align*}
Similar to (\ref{om=}) -- (\ref{rnom=}) it can be shown that
\begin{align*}
\overline{D}_{n}(\tau,t_{4},t_{5})&=\overline{v}_n(\tau)\overline{\eta}_{n}^{C}(t_{4},t_{5})+O(n^{-1/2}),\quad n\rightarrow\infty,
\end{align*}
where $\eta_n^C$ is defined in (\ref{eta}), so that by (\ref{Vareta})
 \begin{align}
\lim_{n\rightarrow\infty}\overline{D}_{n}(\tau,t_{4},t_{5})&=
2K_{A}^{'(2)}v(\tau)v(t_{4})v(t_{5}).\label{D=}
\end{align}
We also have for $B_n$:
\begin{align}
\overline{B}_{n}(t_{3},t_{4},t_{5})&=\overline{v}_n(t_{3})n^{-3/2}\sum_{l,m=1}^n\mathbf{E}\Big\{A^{(n)}_{mm}
(U^{4}C^{(n)}U^{5})_{lm}\Big\}+R_{n},\label{Bn=}
\end{align}
where
\begin{align*}
|R_{n}|\leq n^{-1}||A^{(n)}||\Big(\mathbf{E}\Big\{\sum_{m=1}^n|(AU^{3})_{mm}^\circ|^2\Big\}\Big)^{1/2}.
\end{align*}
By the standard argument based on  Poincar\'{e} inequality (\ref{Nash}) one
can easily get \begin{align*}
\mathbf{Var}\{(AU)_{mm}(t)\} &\leq {C|t|^2}{n^{-1}}(AA^{(n)T})_{mm},
\end{align*}
hence,
\begin{align}
R_n=O(n^{-1/2}),\quad {n\rightarrow\infty}.\label{BR}
\end{align}
Besides, repeating with obvious modifications steps leading from (\ref{etaeq})
to (\ref{etaK}), we get
\begin{align*}
\lim_{n\rightarrow\infty}n^{-3/2}\sum_{l,m=1}\mathbf{E}\Big\{A^{(n)}_{mm}
(U^{4}C^{(n)}U^{5})_{lm}\Big\}=K^{(1)}_{A}v(t_{4})v(t_{5}).
\end{align*}
This, (\ref{Bn=}), and (\ref{BR})
yield
\begin{align}
B(t_{3},t_{4},t_{5}):=\lim_{n\rightarrow\infty}\overline{B}_{n}(t_{3},t_{4},t_{5})=K^{(1)}_{A} v(t_{3})v(t_{4})v(t_{5}).\label{B=}
\end{align}
Plugging (\ref{D=}), (\ref{B=}) in (\ref{Gam=}) we get equation with respect
to $\Gamma=\lim_{n\rightarrow\infty}\overline{\Gamma}_{n}$:
\begin{align*}
\Gamma(t_{2},t_{3},t_{4},t_{5})&+{w^{2}}
\int_{0}^{t_{2}}dt_{6}\int_{0}^{t_{6}}v(t_{6}-t_{7})\Gamma(t_{7},t_{3},t_{4},t_{5})dt_{7}
\\
&=B(t_{3},t_{4},t_{5})
-{w^{2}}K^{(2)}_{A}v(t_{4})v(t_{5})
\int_{0}^{t_{2}}dt_{6}\int_{0}^{t_{3}}v(t_{6}+t_{7})v(t_{3}-t_{7})dt_{7},\notag
\end{align*}
where we put $K^{(2)}_{A}={2}T_AK^{'(2)}_{A}$ (see (\ref{K2}), (\ref{K22})). Solving the equation with the help of (\ref{F1}) -- (\ref{F2}), we obtain
\begin{align*}
\Gamma(t_{2},t_{3},t_{4},t_{5})&=v(t_{2})B(t_{3},t_{4},t_{5})+K^{(2)}_{A}(v(t_{2}+t_{3})-v(t_{2})v(t_{3}))
v(t_{4})v(t_{5})
\\
&=(K^{(1)}_{A}-K^{(2)}_{A})\prod_{j=1}^5v(t_j)+K^{(2)}_{A}v(t_{2}+t_{3})
v(t_{4})v(t_{5}).
\end{align*}
This and (\ref{om=}) -- (\ref{rnom=})
finally yield (\ref{limw}).

\medskip

{\bf (vi)} It follows from Poincar\'{e} inequality (\ref{Nash})
that
\begin{align*}
\mathbf{Var}\{\gamma^{(1)}_{n}(\overline{t^{(2p+2)}}) \} &\leq\frac{w^{2}}{n^{(p+2)}} \sum_{1\leq j \leq k \leq
n} \beta^{-1}_{jk}\mathbf{E}\Big\{\Big|D_{jk}\sum_{ l, m=1}^n(U^{1}A^{(n)}U^{2})_{lm}\prod_{j=2}^{p+1} (U^{2j-1}C^{(n)}U^{2j})_{lm}\Big|^2\Big\}.
\end{align*}
Taking into account  (\ref{DUAU}) we see that to get the first part of (\ref{limg1}) it
suffices to show that
\begin{align*}
R_{n}:=\frac{1}{n^{(p+2)}} \sum_{ j, k=1}^n \Big|\sum_{ l, m=1}^nU^{0}_{jl}(U^{1}A^{(n)}U^{2})_{km}\prod_{j=2}^{p+1} (U^{2j-1}C^{(n)}U^{2j})_{lm}\Big|^2=O(n^{-1}),
\end{align*}
as $n\rightarrow\infty$, (here $U^0=U(t_0)$). Since by (\ref{norU})
\begin{align*}
 \sum_{ j=1}^n U_{jl}\overline{U_{jl'}}=\delta_{ll'},\quad \sum_{ k=1}^n
 (U^{2}A^{(n)T}U^{1})_{mk}\overline{(U^{1}A^{(n)}U^{2})_{km'}}=\sum_{ k=1}^n
 (U^{2}A^{(n)T})_{mk}\overline{(A^{(n)}U^{2})_{km'}},
\end{align*}
then
\begin{align*}
R_{n}&=\frac{1}{n^{(p+2)}} \sum_{  k,l=1}^n \Big|\sum_{  m=1}^n(A^{(n)}U^{2})_{km}\prod_{j=2}^{p+1}
(U^{2j-1}C^{(n)}U^{2j})_{lm}\Big|^2
\\
&\leq\frac{1}{n^{(p+2)}} \sum_{  k,l=1}^n \sum_{  m=1}^n|(A^{(n)}U^{2})_{km}|^2
\sum_{  m'=1}^n\Big|\prod_{j=2}^{p+1} (U^{2j-1}C^{(n)}U^{2j})_{lm'}\Big|^2. \end{align*}
We have by (\ref{1}) and (\ref{norU})
\begin{align*}
 \sum_{k,m=1}^n|(A^{(n)}U^{2})_{km}|^2=\Tr A^{(n)}A^{(n)T} =O(n),\quad n\rightarrow\infty,
 \end{align*}
 and by (\ref{UAUOn}) -- (\ref{UAUlm})
  \begin{align*}
\sum_{  l,m'=1}^n\Big|\prod_{j=2}^{p+1} (U^{2j-1}C^{(n)}U^{2j})_{lm'}\Big|^2=O(n^{p-1})\sum_{  l,m'=1}^n\Big| (U^{3}C^{(n)}U^{4})_{lm'}\Big|^2=O(n^{p}),\quad n\rightarrow\infty. \end{align*}
Hence, $R_n=O(n^{-1})$, $n\rightarrow\infty$, and
 \begin{align}
\mathbf{Var}\{\gamma^{(1)}_{n}(\overline{t^{(2p+2)}}) \} &=O(n^{-1}),
\quad n\rightarrow\infty.\label{varg1}
 \end{align}
 To prove (\ref{limg1}) we show that every $U(t)$ in
\begin{align}
\overline{\gamma}^{(1)}_{n}(\overline{t^{(2p+2)}})=n^{-(p+1)/2}\sum_{l,m=1}^n
\mathbf{E}\Big\{(U(t_{1})A^{(n)}U(t_{2}))_{lm}\prod_{j=2}^{p+1} (U(t_{2j-1})C^{(n)}U(t_{2j}))_{lm}\Big\}
\label{eg1}
 \end{align}
 can be
 replaced with $\overline{v}_n$ with the error term that vanishes as   $n\rightarrow\infty$.
For this purpose it suffices to show that
\begin{align}
\overline{\gamma}^{(1)}_{n}(\overline{t^{(p)}})=\overline{v}_n(t_{1})\overline{\delta}_n(t_2,...,t_{2p+2})+o(1), \quad  n\rightarrow\infty,\label{g1=}
 \end{align}
 where
\begin{align}
{\delta}_n(t_2,...,t_{2p+2})=n^{-(p+1)/2}\sum_{l,m=1}^n
(A^{(n)}U(t_{2}))_{lm}\prod_{j=2}^{p+1} (U(t_{2j-1})C^{(n)}U(t_{2j}))_{lm}.\label{del}
 \end{align}
 Applying  Duhamel formula (\ref{Duh}) and then
differentiation formulas (\ref{diffga}), (\ref{DAU}) -- (\ref{DUAUlm}), we get:
 \begin{align*}
\overline{\gamma}^{(1)}_{n}(\overline{t^{(p)}})&=\overline{\delta}_n(t_2,...,t_{2p+2})
\\
&-{w^{2}}
\int_{0}^{t_{1}}d\tau_{1}\int_{0}^{\tau_{1}}v(\tau_{1}-\tau_{2})\overline{\gamma}^{(1)}_{n}(\tau_{2},t_{2},...,t_{2p+2})d\tau_{2}\notag \\
&-{w^{2}}
\int_{0}^{t_{1}}R_{n}(\tau_{1},t_{2},...,t_{2p+2})d\tau_{1},\notag
\end{align*}
where
 \begin{align}
R_{n}(\tau_{1},&t_{2},...,t_{2p+2})=\int_{0}^{\tau_{1}}\mathbf{E}\big\{v^{\circ}_{n}(\tau_{2})
{\gamma }^{(1)}_{n}(\tau_{1}-\tau_{2},t_{2},...,t_{2p+2})\big\}d\tau_{2}\label{Rgn=} \\
&+\int_{0}^{\tau_{1}}(\overline{v}_n(\tau_{2})-v(\tau_{2}))
\overline{\gamma }^{(1)}_{n}(\tau_{1}-\tau_{2},t_{2},...,t_{2p+2})d\tau_{2}\notag
\\
&+n^{-1}\tau_{1}\overline{\gamma }^{(1)}_{n}(\tau_{1},t_{2},...,t_{2p+2})\notag
\\
&+n^{-1}
\int_{0}^{t_{2}}\overline{\gamma }^{(1)}_{n}(\tau_{1}+\tau_{2},t_{2}-\tau_{2},...,t_{2p+2})
d\tau_{2}\notag
\\&+
\int_{0}^{t_{2}}
\mathbf{E}\Big\{n^{-1}\xi_{n}^{A}(\tau_{1}+\tau_{2})\frac{1}{n^{(p+1)/2}}\sum_{l,m=1}^nU_{lm}(t_{2}-\tau_{2})
\prod_{j=2}^{p+1} (U(t_{2j-1})C^{(n)}U(t_{2j}))_{lm}\Big\}
d\tau_{2}\notag
\\
&+\frac{1}{n}\cdot
\frac{1}{n^{(p+1)/2}}\sum_{l,m,k=1}^n\beta^{-1}_{lk}\mathbf{E}\Big\{(U(\tau_{1})A^{(n)}U(t_{2}))_{km}
D_{lk}\prod_{j=2}^{p+1} (U(t_{2j-1})C^{(n)}U(t_{2j}))_{lm}\Big\}.\notag
\end{align}
It follows from (\ref{F1}) and  (\ref{solut}) with $T=-v$ and
$R'(t)=-w^{2}R_n(t,t_{2},...,t_{2p+2})$, that
 \begin{align*}
\overline{\gamma}^{(1)}_{n}(\overline{t^{(p)}})&=v(t_1)\overline{\delta}_n(t_2,...,t_{2p+2})
-{w^{2}}
\int_{0}^{t_{1}}v(t_1-\tau_{1})R_{n}(\tau_{1},t_{2},...,t_{2p+2})d\tau_{1}.\\
\notag
\end{align*}
Hence, to get (\ref{g1=}) it suffices to show that
\begin{align}
R_{n}=o(1), \quad  n\rightarrow\infty.\label{Rn=o}
 \end{align}
Indeed, the first four terms of the r.h.s. of (\ref{Rgn=}) vanishes because of (\ref{xiv}), the fifth term is of the order $O(n^{-1/2})$, $n\rightarrow\infty,$
because of (\ref{UAUOn}) -- (\ref{UAUlm}) and boundedness of $n^{-1}\xi_{n}^{A}(\tau_{1}+\tau_{2})$,
and the last term after differentiation gives terms of the form $n^{-1}\overline{\gamma }^{(1)}_n$ or
\begin{align}
\frac{1}{n}\cdot
\frac{1}{n^{(p+1)/2}}\sum_{l,m=1}^n\mathbf{E}\Big\{(UA^{(n)}UC^{(n)}U)_{lm}\prod_{j=2}^{p+1} (UC^{(n)}U)_{lm}\Big\},
 \end{align}
 which evidently of the order $O(n^{-1/2})$, $n\rightarrow\infty$ (see (\ref{UAUOn}) -- (\ref{UAUlm}) ). Hence, (\ref{Rn=o}) is proved, and so does (\ref{g1=}).
It remains to note that  (\ref{g1=}) holds true for
\begin{align*}
\overline{\gamma}^{(1)}_{n}(\overline{t^{(2p+2)}})=n^{-(p+1)/2}\sum_{l,m=1}^n
\mathbf{E}\Big\{(V(t_{1})A^{(n)}V(t_{2}))_{lm}\prod_{j=2}^{p+1} (V(t_{2j-1})C^{(n)}V(t_{2j}))_{lm}\Big\}.
 \end{align*}
(cf (\ref{eg1})), where $V$ is equal $U$ or identity matrix $I_n$. Hence,
in the limit $n\rightarrow\infty$
we can  replace all $U$ of  (\ref{eg1}) with $v$ and so get (\ref{limg1}).

{\bf (vii)} The proof of (vii) repeats essentially that one of  (vi).

\medskip
{\bf Wigner case.} Proofs of all statements (i) -- (vii) follow the same
scheme based on the known facts for the GOE matrices and interpolation procedure
proposed while proving Theorems \ref{t:Cov} and \ref{t:clt}. We demonstrate
this scheme proving (i):

{\bf (i)} Consider $V_{n}(t):=\mathbf{Var}\{\xi_{n}^{A}(t)\}$ and note that
\begin{align}
V_{n}(t)&=\mathbf{Var}\{\widehat{\xi}_{n}^{A}(t)\} +C^{\Delta}_{n}(t,-t),\label{Vn=}
\end{align}%
where $\widehat{\xi}_{n}^{A}$ and $C^{\Delta}_{n}$ are defined in (\ref{UGOE})
and (\ref{CDelta}), respectively. By
(\ref{vxi<}) we have
\begin{align}
\mathbf{Var}\{\widehat{\xi}_{n}^{A}(t)\}\leq ct^2.\label{V1}
\end{align}%
 Repeating steps leading from  (\ref{CDelta}) to (\ref{c}) -- (\ref{e2<}),
 but using here (\ref{difgen}) with $p=5$ instead $p=6$ in (\ref{c}), we get
\begin{align}
c^{\Delta}_{n}(t,-t)&=\frac{i}{2}\int_{0}^{1}\Big[\sum_{j=2}^5 {s^{(j-1)/2}}T_j^{(n)}+\varepsilon_5\Big]%
ds  \label{ct-t}
\end{align}
with $T_j^{(n)}$ of (\ref{Tp}), and
\begin{equation}
|\varepsilon_{5}|\leq\frac{C_5w^{7/8}_8}{n^{7/2}}\sum_{l,m=1}^n\sup_{M\in \mathcal{S}%
_{n}}|D_{lm}^{6}\Phi_{lm}|\leq {c(1+|t_{}|)^7}{}.  \label{e5<}
\end{equation}
Consider  $T_1^{(n)}$. It is given by (\ref{lT2=}) with $t_1=t$, $t_2=-t$. Since  $T_{21}^{j(n)}$, $j=1,2,3$ of   (\ref{T211}) -- (\ref{T213}) are bounded
uniformly in $n\in \mathbb{N}$, and  every derivative $D_{lm}$
of $U(t)=e^{itM^{(n)}}$ gives factor $t$, then
 \begin{align*}
\Big|n^{-3/2}\sum_{l,m=1}^nD_{lm}^{2}(U*A^{(n)}U)_{lm}\Big|\leq c(1+|t|)^{3},
\end{align*}%
and by the Schwarz inequality we have for $T_{21}^{(n)}$ of (\ref{lT2=}):
 \begin{align*}
|T_{21}^{(n)}|=\Big|\mathbf{E}\Big\{n^{-3/2}\sum_{l,m=1}^nD_{lm}^{2}(U*A^{(n)}U)_{lm}\cdot\xi_{n}^{A\circ}(t)\Big\}\Big|\leq c(1+|t|)^{3}V_{n}^{1/2}.
\end{align*}%
We also have for $T_{22}^{(n)}$ and $T_{23}^{(n)}$of (\ref{lT2=})
(see  (\ref{T2223}) and  (\ref{UAUlm}) -- (\ref{UAUmm})):
 \begin{align*}
|T_{22}^{(n)}+T_{23}^{(n)}|\leq c(1+|t|)^{3}.
\end{align*}
Hence,
 \begin{align}
|T_{2}^{(n)}|\leq c(1+|t|)^{3}(V_{n}^{1/2}+1).\label{VT2}
\end{align}%
Treating $T_{3}^{(n)}$of (\ref{T3})
and $T_{j}^{(n)}$, $j=4,5$ of (\ref{Tp}) in the similar way one can get
 \begin{align}
&|T_{3}^{(n)}|\leq c(1+|t|)^{4}(V_{n}^{1/2}+1),\label{VT3}
\\
&|T_{j}^{(n)}|\leq c(1+|t|)^{j+1},\quad j=4,5.\label{VTj}
\end{align}%
Putting (\ref{e5<}) -- (\ref{VTj})
in (\ref{ct-t}), and then together with (\ref{V1}) in (\ref{Vn=}), we get
 the quadratic inequality with respect to $V_{n}^{1/2}$:
\begin{equation*}
V_{n}-c(1+|t|)^4V_{n}^{1/2}-c(1+|t|)^7\leq 0,
\end{equation*}
solving which we get $V_{n}\leq c(1+|t|)^4$.

To finish the proof of (i) it remains to show that
\begin{equation}  \label{Uv}
\lim_{n\rightarrow\infty}n^{-1}\mathbf{E}\{\xi_{n}^{A}(t)\}=T_{A}v(t).
\end{equation}
In the GOE case we have (see (\ref{xiTA}))
\begin{equation}  \label{UvG}
\lim_{n\rightarrow\infty}n^{-1}\mathbf{E}\{\widehat{\xi}_{n}^{A}(t)\}=T_{A}v(t).
\end{equation}
Besides, we have
\begin{align*}
\xi_{n}^{A}(t)-\widehat{\xi}_{n}^{A}(t)&=\int_{0}^{1}\frac{\partial }{\partial s}%
\xi_{n}^{A}(t,s) ds   \\
&=\frac{i}{2}\int_{0}^{1}\sum_{l,m=1}^n \Big(\frac{1}{\sqrt{sn}}W^{(n)}_{lm}-%
\frac{1}{\sqrt{(1-s)n}}\widehat{W}_{lm}\Big)(U_{}*A^{(n)}U_{})_{lm}(t,s)ds,  \notag
\end{align*}
so that similar to (\ref{ct-t}) -- (\ref{e5<})
\begin{align*}
n^{-1}\mathbf{E}\{\xi_{n}^{A}(t)\}-n^{-1}\mathbf{E}\{\widehat{\xi}_{n}^{A}(t)\}=\frac{i}{2}\int_{0}^{1}
\Big[ {s^{1/2}}T_2^{\prime(n)}+\varepsilon_2\Big]%
ds,
\end{align*}
where
\begin{equation*}
T_2^{\prime(n)}=\frac{\kappa _{3}}{j!n^{5/2}}\sum_{l,m=1}^{n}\mathbf{E}%
\big\{D_{lm}^{2}(U_{}*A^{(n)}U_{})_{lm}(t,s)\big\}=O(n^{-1}),\quad n\rightarrow\infty,
\end{equation*}
and
\begin{equation*}
|\varepsilon_{3}|\leq\frac{C_3\sqrt{w^{}_8}}{n^{6}}\sum_{l,m=1}^n\sup_{M\in \mathcal{S}%
_{n}}|D_{lm}^{4}(U_{}*A^{(n)}U_{})_{lm}(t,s)|=O(n^{-1/2}),\quad n\rightarrow\infty.
\end{equation*}
Hence,
\begin{align*}
n^{-1}\mathbf{E}\{\xi_{n}^{A}(t)\}-n^{-1}\mathbf{E}\{\widehat{\xi}_{n}^{A}(t)\}=O(n^{-1/2}),\;n\rightarrow\infty.
\end{align*}
This and (\ref{UvG}) yield (\ref{Uv}) and finish the proof of (i).
\end{proof}

\end{document}